\documentclass[journal,onecolumn]{IEEEtran}
\usepackage{amsthm}
\usepackage{amsfonts}
\usepackage{amsmath}
\usepackage{mathrsfs}
\usepackage{soul,color}
\usepackage{xcolor}
\usepackage{cite}
\usepackage{enumitem}   
\usepackage{graphicx}
\usepackage{siunitx}
\usepackage{multirow}
\usepackage{algorithm}
\usepackage{algorithmicx}  
\usepackage{algpseudocode}
\usepackage{booktabs}
\renewcommand{\algorithmicrequire}{\textbf{Input: }}

\newcommand{\End}{\hfill $\square$}
\DeclareMathOperator*{\argmin}{argmin}

 \DeclareMathOperator*{\trace}{trace}

\newcommand{\mbb}[1]{\mathbb{#1 }}
\newcommand{\mbf}[1]{\mathbf{#1}} 
\newcommand{\mcl}[1]{\mathcal{#1}}

\newcommand{\pce}[1]{\textsf{#1}}
\newcommand{\pcecoe}[2]{\textsf{#1}^{#2}}
\newcommand{\relx}{(\omega)}

\newcommand{\trar}[2]{\mathbf{#1}_{[0,{#2}]}}
\newcommand{\set}[1]{\mathbb{#1}}    

\newcommand{\inst}[1]{_{#1}}   
\newcommand{\pred}[2]{_{{#1}|{#2}}}

\newcommand{\splx}[1]{\mcl{L}^2(\Omega, \mathcal{F}, \mu; \mathbb{R}^{#1})}
\newcommand{\spl}{\mcl{L}^2(\Omega, \mathcal{F}, \mu; \mathbb{R})} 
 
\newcommand{\diff}{\mathop{}\!\mathrm{d}}

\newcommand{\mean}{\mbb{E}}
\newcommand{\var}{\mbb{V}}
\newcommand{\Ucov}{ \widehat{\Sigma}}
 \newcommand{\covar}{\Sigma}

\newcommand{\prob}{\mbb{P}}
\newcommand{\Hankel}{\mcl{H}}

\newcommand{\rinv}{\dagger}

\newcommand{\ini}{_{\text{ini}}} % initial condition

\newcommand{\dimx}{{n_x}}
\newcommand{\dimu}{{n_u}}
\newcommand{\dimw}{{n_x}}

\newcommand{\I}{\mathbb{I}}
\newcommand{\N}{\mathbb{N}}
\newcommand{\R}{\mathbb{R}}
\newcommand{\Xf}{\mathbb{X}_\text{f}}
\newtheorem{theorem}{Theorem}
\newtheorem{corollary}{Corollary}
\newtheorem{lemma}{Lemma}
\newtheorem{definition}{Definition}
\newtheorem{remarkmod}{Remark}
\newtheorem{example}{Example}
\newtheorem{assumption}{Assumption}

\newcommand{\xxd}{ \Hankel_1\left(\mbf{x}_{[0,T-1]}\right)}
\newcommand{\xxdp}{ \Hankel_1\left(\mbf{x}_{[1,T]}\right)}
\newcommand{\uud}{ \Hankel_1\left(\mbf{u}_{[0,T-1]}\right)}
\newcommand{\wwd}{ \Hankel_1\left(\mbf{w}_{[0,T-1]}\right)}

\newcommand{\edit}[1]{{\color{black} #1}}

\IEEEoverridecommandlockouts                              % This command is only needed if 

\usepackage{amsmath} % assumes amsmath package installed
\usepackage{amssymb}  % assumes amsmath package installed

\title{Towards data-driven stochastic predictive control}

\author{Guanru Pan, Ruchuan Ou and Timm Faulwasser$^{\star}$% <-this % stops a space
	\thanks{ $^{\star}$: Corresponding author.}%
	\thanks{Guanru Pan, Ruchuan Ou and Timm Faulwasser are with Institute for Energy Systems, Energy Efficiency and Energy Economics, TU Dortmund University, Dortmund,  Germany
		{\tt\small $\{$guanru.pan,ruchuan.ou$\}$@tu-dortmund.de and timm.faulwasser@ieee.org} }
}

\begin{document}
\maketitle
\thispagestyle{empty}
\pagestyle{empty}

\begin{abstract}
 Data-driven predictive control based on the fundamental lemma by Willems et al.  is frequently considered for deterministic LTI systems subject to measurement noise.  However, little has been done on data-driven stochastic control. In this paper, we propose a data-driven stochastic predictive control scheme for LTI systems subject to possibly unbounded additive process disturbances. Based on a stochastic extension of the fundamental lemma and leveraging polynomial chaos expansions,  we construct a data-driven surrogate Optimal Control Problem (OCP). Moreover, combined with an online selection strategy of the initial condition of the OCP, we provide sufficient conditions for recursive feasibility and for stability of the proposed data-driven predictive control scheme. Finally, two numerical examples illustrate the efficacy and closed-loop properties of the proposed scheme for process disturbances governed by different distributions.
\end{abstract}

% Note that keywords are not normally used for peerreview papers.

\section{INTRODUCTION}
The so-called fundamental lemma proposed by Willems et al. in 2005~\cite{Willems2005} is at the core of manifold research efforts on data-driven control~\cite{DePersis19,Coulson2019,Berberich20,Pan21s}.  Its main insight is that the trajectories of any controllable LTI system can be described without explicit parametrization or derivation of a state-space model. Instead, all finite-length input-output trajectories of LTI systems are shown \edit{lying}  in the column space of a Hankel matrix constructed from recorded input-output data. Provided persistency of excitation holds, this data-driven system representation is exact for deterministic LTI systems, \edit{i.e. it is exact in the absence of measurement noise or process disturbances}.

The use of Willems' fundamental lemma for predictive control dates back to the work by Yang and Li~\cite{Yang15}, while later it has been popularized by Coulson et al.~\cite{Coulson2019}. Data-driven predictive control schemes are promising for different applications~\cite{Huang19,Lian21,Berberich21at,Carlet20,schmitz22a}.  So far available results on formal  properties, e.g., distributional robustness~\cite{Coulson21a}, recursive feasibility, and stability~\cite{Berberich20}, are restricted to deterministic LTI systems with and without measurement noise. Besides measurement noise, however, stochastic process \edit{disturbances} might also directly affect the dynamics. In the present paper, we aim to extend deterministic data-driven predictive control towards stochastic LTI systems affected by possibly unbounded additive process disturbances. To this end, we combine two lines of research, i.e., we rely on transferring known model-based results for stochastic predictive control to the data-driven setting. In turn, the data-driven system description is approached via a stochastic extension to the fundamental lemma, which we proposed in an earlier paper~\cite{Pan21s} .

In stochastic predictive control one major challenge is the propagation of uncertainties through system dynamics~\cite{Mesbah16,Farina16,Heirung18}. Polynomial Chaos Expansion (PCE) is a computationally tractable method for uncertainty propagation in Markovian and non-Markovian settings. Its core idea is based on the observation that, under mild assumptions, random variables can be regarded as elements of an $\mcl{L}^2$  probability space. Hence they admit representations in appropriately chosen polynomial bases, see the monograph by Sullivan~\cite{sullivan15introduction} for a general introduction. Moreover, PCE is established in systems and control, see the work by Fagiano et al.~\cite{Fagiano12} and others~\cite{kim13wiener,paulson14fast,Mesbah14}.  It also allows for efficient computation of  statistical moments~\cite{Lefebvre20}.

Another challenging problem in stochastic predictive control is to establish recursive feasibility of the underlying Optimal Control Problem (OCP). Early works rely on the boundedness of the disturbances to establish recursive feasibility~\cite{Cannon10s,Kouvaritakis10e,Korda11s}. For process \edit{disturbances} with unbounded support the underlying OCP can become infeasible in case of large disturbance realizations. Therefore, backup strategies to recover feasibility are  of interest: Cannon et al.~\cite{Cannon09p} achieve feasibility by solving an alternative OCP that minimizes the constraint violation; while Farina et al.~\cite{Farina13p,Farina15a}  solve the original OCP with a modified feasible initial condition. 

In a recent paper~\cite{Pan21s}, we combine the PCE approach to uncertainty propagation with a data-driven non-parametric system representation via the fundamental lemma. Specifically, therein  we \edit{present} a stochastic variant of the fundamental lemma which uses PCE to achieve non-parametric system representation and uncertainty propagation for stochastic LTI systems subject to additive Gaussian or non-Gaussian process \edit{disturbances}. The stochastic fundamental lemma enables propagating uncertainty induced by process disturbances through system dynamics and predicting the statistical distributions %of inputs and
of outputs over finite horizons \edit{given the distributions of inputs and disturbances}. This is accomplished without explicitly identifying \edit{a model of the underlying system}. Based on the stochastic fundamental lemma, \edit{we present} an intuitive structure for stochastic data-driven predictive control with promising numerical results~\cite{Pan21s}. However, therein we \edit{do} not provide any formal analysis of the closed-loop properties.
In an alternative route towards data-driven stochastic control, Kerz et al.~\cite{Kerz21d} exploit the deterministic fundamental lemma to stochastic predictive control with a tube-based approach. That is, they predict the nominal part of state trajectories by the fundamental lemma while handle the uncertainties by tightening constraints offline.   

This paper proposes a data-driven predictive control scheme for stochastic LTI systems subject to possibly unbounded additive process \edit{disturbances} with feasibility guarantees. Our main contribution are sufficient conditions for stability  and recursive feasibility of the proposed data-driven predictive control algorithm. 

The remainder of the paper is structured as follows: Section~\ref{sec:problem_statement} details settings of the considered stochastic predictive control problem and preliminaries on the data-driven representation of stochastic LTI systems. Section~\ref{sec:OCPS} presents the main results, i.e., based on the data-driven reformulation of the stochastic OCP we analyse the closed-loop properties of the proposed control scheme. Section~\ref{sec:simulation} considers two numerical examples; the paper ends with conclusions in Section~\ref{sec:conclusion}.

\subsubsection*{Notation}
Let $Z: \I_{[0,T-1]} \rightarrow \splx{n_z}$ be a vector-valued random variables and $\splx{n_z}$ is the underlying probability space. 
We denote by $\mean\big[Z\big]$, $\var\big[Z\big]$, $\covar\big[Z\big]$, and $z\doteq Z(\omega)\in \R^{n_z}$ its mean, variance, covariance, and realization, respectively.  Moreover, for a given set $\set{Z}\subseteq \R^{n_z}$, we denote the probability of $Z \in \set{Z}$ as $\prob\big[Z\in \set{Z}\big]$. The vectorization of a sequence $\{Z_k\}_{k=0}^{T-1}$ is written as $\trar{Z}{T-1} \doteq [Z_0^\top,Z_1^\top, \dots,Z_{T-1}^\top]^\top \in \R^{n_z T}$. The PCE coefficient sequence of $\{Z_k\}_{k=0}^{T-1}$  is  denoted by $\pcecoe{z}{j}:  \I_{[0,T-1]} \rightarrow \R^{n_z},j\in \N$. Similarly, we write the vectorization of $\{z_k\}_{k=0}^{T-1}$ and $\{\pcecoe{z}{j}_k\}_{k=0}^{T-1}$ as $\trar{z}{T-1}$ and $\trar{\pce{z}}{T-1}^j$.
Throughout the paper, we denote the identity matrix of size $l$ by $I_l$ and $\|x\|^2_Q \doteq \frac{1}{2}  x^\top Q x$. For any matrix $Q\in \mbb{R}^{n\times m}$ with columns $q^1, \dots, q^m$, we use $\mathrm{colsp}(Q) \doteq \mathrm{span}\left(\{q^1,\dots, q^m\}\right)$.

\section{Preliminaries} \label{sec:problem_statement}
Next, we introduce the considered setting of stochastic LTI systems, recall known results for model-based stochastic MPC, and then revisit the polynomial chaos framework for uncertainty propagation. Finally, we recapitulate the non-parametric description of  stochastic LTI systems via a stochastic extension to the fundamental lemma.

\subsection{Stochastic LTI systems}
We consider stochastic discrete-time LTI systems of the form
\begin{equation}\label{eq:RVdynamics}
X\inst{k+1} = AX\inst{k} +BU\inst{k}+ W\inst{k}, \quad 
X_0=X\ini, \quad k \in \N
\end{equation}
with state  $X\inst{k}\in \mcl L^2(\Omega, \mcl F_k,\mu;\R^{n_x})$,  input $U\inst{k}\in \mcl L^2(\Omega, \mcl F_k,\mu;\R^{n_u})$ and  process \edit{disturbances} $W\inst{k}\in \mcl L^2(\Omega, \mcl F,\mu;\R^{n_x})$, $k \in \N$, where $\Omega$ is the sample space, $\mcl F$ is a $\sigma$-algebra, $\mbb F\doteq (\mcl F_k)_{k \in \N}$ is a stochastic filtration, and $\mu$ is the considered probability measure. The consideration of $\mcl L^2$ random variables implies  that their expectation and covariance are finite. In the underlying filtered probability space $\mcl L^2(\Omega, \mcl F, \mbb F, \mu)$, the $\sigma$-algebra contains all available historical information, or more precisely,
\begin{equation}\label{eq:filtration}
\mcl F_0 \subseteq \mcl F_1 \subseteq ...  \subseteq \mcl F.
\end{equation}
Let $\mbb F$ be the smallest filtration that the stochastic process $X$ is adapted to, i.e.,
$
\mcl F_k = \sigma(X\inst{i},i\leq k)$,
where $\sigma(X\inst{i},i\leq k)$ denotes the $\sigma$-algebra generated by $X\inst{i},i\leq k$. For more details on filtrations we refer to the book of Fristedt and Gray~\cite{fristedt13modern}.
Likewise, the stochastic input $U\inst{k}$ is also modelled as a stochastic process that is adapted to the filtration $\mbb F$, i.e., $U\inst k$ may only depend on the information available up to time $k$, i.e. $X\inst{0}, X\inst{1},...,X\inst{k}$. Note that the influence of the process disturbances $W\inst i,i< k$ is implicitly included via the state recursion~\eqref{eq:RVdynamics}.

Throughout the paper, the underlying probability distribution of the initial condition $X_{\text{ini}}$  are assumed to be known, and we consider $W_k$, $ k\in \N$ to be zero-mean identical independently distributed ($i.i.d.$) random variables with known distributions. 

\begin{remarkmod}[Realization dynamics]
 For a specific uncertainty outcome $\omega \in \Omega$, the realization of $W\inst{k}$ is written as $w\inst{k} \doteq W\inst{k}(\omega) $. Likewise, we denote state and input realizations as $x\inst{k} \doteq X\inst{k}(\omega)$ and $u\inst{k} \doteq U\inst{k}(\omega)$ for $k\in\N$, respectively. 
 
 Moreover, given  a specific initial condition $x\ini \doteq X\ini(\omega)$ with $\omega \in \Omega$,
 and a sequence of disturbance realizations $w_k$, $k\in\N$,  the stochastic system \eqref{eq:RVdynamics} induces the \textit{realization dynamics}
\begin{equation}\label{eq:RelizationDynamics}
x\inst{k+1} = Ax\inst{k} +Bu\inst{k}+ w\inst{k}, \quad  
x\inst{0}=x\ini,\quad k \in \N
\end{equation}
which are satisfied by state, input, and \edit{disturbance} realization sequences.\End
\end{remarkmod}
\begin{assumption}[Standing assumptions]\label{ass:AB}  Consider stochastic LTI system \eqref{eq:RVdynamics} and its realization dynamics \eqref{eq:RelizationDynamics}, we assume $(A,B)$ is a controllable pair, and we suppose that the matrices $A$, $B$ as well as the \edit{disturbance} realizations $w_k$ are unknown, while past state realizations $x_k$ and input realizations $u_k$ are assumed to be known/measured.\End
\end{assumption}

\subsection{Model-based stochastic OCP with terminal constraints}
Next, we recall a conceptual framework for stochastic model-based MPC proposed by Farina et al.~\cite{Farina13p} and others~\cite{Cannon09p,Magni09s}. This scheme \edit{serves} as a blueprint for our further contributions, i.e., the data-driven design of stochastic MPC.

At time instant $k$, given the probabilistic initial condition $\bar{X}_k$ and i.i.d. zero-mean \edit{disturbances} $W\inst{i}, i\in \I_{[k,k+N-1]}$, we consider the following OCP with horizon $N \in \N$
\begin{subequations} \label{eq:stochasticOCP}
\begin{align}
 \min_{
\substack{
\text{for } i \in   \set{I}_{[k,k+N-1]}\\
   X\pred{i+1}{k} \in \splx{\dimx}, \\
   U\pred{i}{k} \in \splx{\dimu}}   
    }  
    \sum_{i=k}^{k+N-1} \mean  \big[ \|X\pred{i}{k}\|^2_Q &+\|U\pred{i}{k}\|^2_R\big] +  \mean \big[ \|X\pred{k+N}{k}\|^2_P\big] \hspace{1pt}  \label{eq:stochasticOCP_obj}\\
\text{subject to }\quad 
X\pred{i+1}{k}=   AX\pred{i}{k}+BU\pred{i}{k}  
+W\inst{i},\label{eq:OCPdynamic}&\quad X\pred{k}{k}  = \bar{X}\inst{k} &\forall i \in \set I_{[k,k+N-1]},\\
\prob \big[X\pred{i}{k}  \in  \set X \big]\geq 1 - \varepsilon_x,  & \quad
\prob \big[U\pred{i}{k} \in  \set U \big]\geq 1 - \varepsilon_u,   &\forall i \in \set I_{[k,k+N-1]},
 \label{eq:chance} \\
\mean \big[X\pred{k+N}{k} \big] \in  \Xf, &\quad \covar \big[X\pred{k+N}{k} \big] \leq \Gamma.  &\quad \label{eq:terminalCons}
\end{align}
\end{subequations}
where $Z\pred{i}{k}$, $Z \in\{X,U,W\}$ denotes the  random variables at $i\in \set I_{[k,k+N-1]}$ predicted at time step $k$. 
We consider a quadratic mean value cost function \eqref{eq:stochasticOCP_obj}  with positive semi-definite $Q= Q^\top = S^\top S\succeq 0$, positive definite $R=R^\top$ and $ P=P^\top \succ 0$. We require detectability of $(A,S)$. Specifically, $P \in \R^{\dimx \times \dimx}$ \edit{is chosen as} the unique positive definite solution to the Discrete-time Algebraic Riccati Equation (DARE)
\begin{subequations}
\begin{equation}\label{eq:DARE}
A^\top P A - P -(A^\top P B)(R+B^\top P B)^{-1}(B^\top P A) +Q = 0
\end{equation}
and the corresponding feedback $K\in \R^{\dimu\times\dimx}$   
\begin{equation}\label{eq:dareK}
K = -(R+B^\top P B)^{-1} B^\top P A.
\end{equation}
\end{subequations}
The feedback $K$   can also be computed via the following optimization problem ~\cite{Feron92n,Dorfler21c}
\begin{subequations}\label{eq:modelbasedP}
\begin{align}
\min_{
\substack{
\tilde{P} \in \R^{\dimx\times\dimx},K\in \R^{\dimx \times \dimu}
}
} &\trace(Q\tilde{P}+K^\top R K  \tilde{P}) \\
\text{subject to }\quad
\tilde{P} \succeq I_{\dimx},&\quad (A+BK)\tilde{P}(A+BK)^\top - \tilde{P} +I_{\dimx} \preceq 0.
\end{align}
\end{subequations}

\edit{The above OCP~\eqref{eq:stochasticOCP}} entails the stochastic system dynamics in~ \eqref{eq:OCPdynamic} with $\bar{X}_k$ as its initial condition. Notice that   $\bar{X}_k$ is modelled as a random variable. As we will see in Section \ref{sec:MPCs}, this gives additional flexibility to address recursive feasibility of the OCP. 
Additionally, we consider chance constraints for the states and the inputs in \eqref{eq:chance}, and the underlying sets $\set{X}\subseteq \set{R}^{\dimx}$ and $\set{U}\subseteq \set{R}^{\dimu}$ are assumed to be closed. Moreover, $1-\varepsilon_x$ and $1-\varepsilon_u$ specify the probabilities with which the---joint in the state dimension but individual in time---chance constraints shall be satisfied. However, notice that---depending on the considered disturbance distribution and the system constraints---the joint consideration of chance constraints for states \textit{and} inputs may jeopardize feasibility of the OCP.

Similar to Farina et al.~\cite{Farina13p}, we consider terminal constraints \eqref{eq:terminalCons} on the mean and covariance of $X\pred{k+N}{k} $.
This ensures the recursive feasibility of the OCP, cf. Section~\ref{sec:MPCs}. Precisely, we consider the terminal region $\Xf \subseteq\R^{\dimx}$ satisfying the following assumption similar to deterministic MPC~\cite{Rawlings20} .

\begin{assumption}[Terminal ingredients]\label{ass:terminalinX} 
Consider a random variable $X\in \splx{\dimx}$ with mean $\mean \big[X\big] \in \R^{\dimx}$ and  covariance $ \covar \big[X\big] \in \R^{\dimx \times \dimx}$. There exists a closed set $\Xf \subseteq \set{X}$ and $0 \in \mathrm{int}(\Xf)$ such that the following conditions hold
\begin{itemize}
\item[(i)]  for all $ \mean\big[X\big] \in \Xf \subseteq \R^{\dimx}, (A+BK)\mean\big[X\big] \in \Xf$,
\item[(ii)] for all $X$ satisfying $\mean \big[X\big] \in  \Xf, \, \covar \big[X\big] \leq \Gamma$, cf.~\eqref{eq:terminalCons}, $\prob\big[X\in \set{X}\big] \geq 1-\varepsilon_x$ and $\prob\big[KX\in \set{U}\big] \geq 1-\varepsilon_u$   hold. \End
\end{itemize}
\end{assumption}

The main difference to the standard assumption of deterministic MPC is that the feedback control input $KX$ is not required to satisfy the input constraints $\mbb{U}$ strictly but only in a probabilistic sense. We note that this assumption can be satisfied choosing a sufficiently small $\gamma >0$ in $\Xf \doteq \{x\in \set{X}| \|x\|_P^2\leq \gamma\}$.
The constraint on the covariance of the terminal state $X\pred{k+N}{k}$ in \eqref{eq:terminalCons}, i.e., the inequality $\covar [X] \leq \Gamma$, is considered element-wise. Moreover, the matrix $\Gamma \in \R^{\dimx \times \dimx}$ is a positive definite and symmetric solution to the  Lyapunov equation
\begin{equation}\label{eq:modelbasedgamma}
 (A+BK)\Gamma(A+BK)^\top -\Gamma+ \Ucov= 0,
\end{equation}
where $\Ucov$ is an upperbound on $\covar[W] \in \R^{\dimw \times \dimw}$.
We remark that \eqref{eq:modelbasedgamma} is related to the steady state of the covariance dynamics~\cite{Farina13p}.

The optimal input sequence to \eqref{eq:stochasticOCP} is denoted as $[U^\star\pred{k}{k},U^\star\pred{k+1}{k},\cdots,U^\star\pred{k+N-1}{k}].$ It is a trajectory of random variables. Hence, conceptually the MPC feedback law acting on the random-variable dynamics~\eqref{eq:RVdynamics} is \edit{to take the first predicted input $U^\star\pred{k}{k} \in \splx{\dimu}$. }
However, it is worth to be remarked that in applications, one is not controlling \eqref{eq:RVdynamics} but rather the  realization dynamics \eqref{eq:RelizationDynamics}. Hence, we postpone the necessary discussion of how to compute the deterministic feedback to Section \ref{sec:MPCs}. We refer to Farina et al.~\cite{Farina13p}
for the \edit{model-based} analysis of closed-loop properties and see Lorenzen et al.~\cite{Lorenzen19} for a stochastic MPC scheme without terminal ingredients.

\subsection{Polynomial chaos expansion of stochastic LTI systems}
After the recap of model-based stochastic predictive control, we now turn towards the computational aspects. To this end, we recall the fundamentals of applying Polynomial Chaos Expansion (PCE) to the stochastic system~\eqref{eq:RVdynamics}. 
PCE as such dates back to Norbert Wiener~\cite{wiener38homogeneous} and a general introduction can be found in the book of Sullivan~\cite{sullivan15introduction}.

 The core idea of PCE is that any $\mathcal{L}^2$ random variable can be expressed in a suitable polynomial basis. To this end, consider an orthogonal polynomial basis $\{\phi^j\}_{j=0}^\infty$ which spans $\spl$, i.e.,
\[
	\langle \phi^i, \phi^j \rangle \doteq \int_{\Omega} \phi^i(\omega)\phi^j(\omega) \diff \mu(\omega) = \delta^{ij}\langle \phi^j,\phi^j\rangle,
\]
where $\delta^{ij}$ is the Kronecker delta. 
The PCE of a real-valued random variable $Z\in \spl$ with respect to the basis $\{\phi^j\}_{j=0}^\infty$ is
\begin{equation}\label{eq:PCE_def}
Z = \sum_{j=0}^{\infty}\pcecoe{z}{j} \phi^j  \text{ with } \pcecoe{z}{j} = \frac{\langle Z, \phi^j \rangle}{\langle \phi^j,\phi^j\rangle},
\end{equation}
where $\pcecoe{z}{j} \in \R$ is called the $j$-th PCE coefficient. 

In numerical implementations the PCE series has to be terminated after a finite number of terms which may lead to truncation errors. For details on truncation errors and error propagation we refer to Field and Grigoriu~\cite{field04accuracy} and Mühlpfordt et al.~\cite{muehlpfordt18comments} Fortunately, random variables that follow some widely used distributions admit exact finite-dimensional PCEs with only two terms in suitable polynomial bases. For Gaussian random variables Hermite polynomials are preferable. For other distributions, we refer to Table~\ref{tab:askey_scheme} for the usual basis choices that allow exact PCEs~\cite{koekoek96askey, xiu02wiener}. Notice that one uses
specific random-variable arguments $\xi \in \splx{n_\xi}$ with $\xi: \Omega \to R^{n_\xi}$ for different polynomial
bases, cf. Table~\ref{tab:askey_scheme}.

\begin{table}[t] 
	\caption{Correspondence of random variables and underlying orthogonal polynomials.}
	\label{tab:askey_scheme}
	\centering
	\begin{tabular}{ccc c}
		\toprule
		Distribution  & Support & Orthogonal basis \edit{$\{\phi^j\}_{j=0}^{\infty}$} & Argument \edit{$\xi\relx$} \\
		\midrule
		Gaussian & $(-\infty, \infty)$ & Hermite & $\mathcal N(0,1)$ \\
		Uniform  & $[a,b]$ & Legendre &$\mathcal U ([-1,1])$ \\
		Beta & $[a,b]$ & Jacobi & $\mathcal{B}(\alpha,\beta,[-1,1])$ \\
		Gamma & $(0,\infty)$ & Laguerre & $\Gamma(\alpha,\beta,(0,\infty))$ \\
		\bottomrule
	\end{tabular}
\end{table}

 A vector-valued random variable $Z\in \splx{n_z}$ is said to admitting an exact finite-dimensional PCE with $L$ terms if
\begin{equation}\label{eq:exactPCE}
Z = \sum_{j=0}^{L-1} \pcecoe{z}{j}\phi^j
\end{equation}
where the $j$-th PCE coefficient of $Z\in\splx{n_z}$ reads
\[
	\pce{z}^j = \begin{bmatrix} \pcecoe{z}{1,j} & \pcecoe{z}{2,j} & \cdots & \pcecoe{z}{n_z,j} \end{bmatrix}^\top.
\]
Moreover,  $\pcecoe{z}{i,j}$ is the $j$-th PCE coefficient of the component $Z^i$ of $Z$. We remark that with a finite PCE, the expected value, variance and covariance of $Z$ can be efficiently calculated from its PCE coefficients~\cite{Lefebvre20} 
\begin{equation}\label{eq:PCEmoments}
\mean\big[Z\big] = \pce{z}^0 \quad \text{,} \quad \var \big[Z\big] = \sum_{j=1}^{L-1} \pce{z}^{j\top}\pce{z}^{j}\langle \phi^j,\phi^j\rangle
\quad \quad \covar \big[Z\big] = \sum_{j=1}^{L-1} \pce{z}^j\pce{z}^{j\top}\langle \phi^j,\phi^j\rangle.
\end{equation}

Based on the knowledge of the distributions of $X_\text{ini}$ and \edit{of} all i.i.d. $W_{k}$, $k\in \N$ we make the following assumption.

\begin{assumption}[Exactness of PCE series]\label{ass:finitenoise}~\\
\begin{itemize}\vspace*{-0.5cm}
\item[(i)] The initial condition $X_\text{ini}$ of \eqref{eq:RVdynamics} and all i.i.d. $W_{k}$, $k\in \N$ admit exact PCEs with finite dimensions $L_{x}$ and $L_w$, i.e.
\[
X_\text{ini}= \sum_{j=0}^{L_x-1}\pce{x}^j \phi_x^j, \qquad W_{k}=  \sum_{j=0}^{L_w-1}\pce{w}^j \phi_{k}^j.
\]
\item[(ii)] The i.i.d. disturbances $W_{k}$ have zero mean for all $k\in \mbb N$ and its (co-)variance is bounded from above by a non-negative \edit{definite} matrix $\Ucov \in \R^{\dimw\times \dimw}$, i.e.
\begin{equation}\label{eq:finitrnoise}
\mbb E\big[W_{k}\big] = \pce{w}_k^0 = 0, \qquad\covar\big[W_{k}\big] = \sum_{j=1}^{L_w-1}\pcecoe{w}{j}\pcecoe{w}{j\top}\langle \phi_{k}^j,\phi_{k}^j\rangle \leq \Ucov,
\end{equation}
where '$\leq$' holds element-wise. 
\End
\end{itemize}
\end{assumption}
Note that the bases $\{\phi_{k}^j\}_{j=0}^{L_w-1}$ of $W_k$ at different time instants $k$ are structurally identical but the realizations used as arguments of the basis functions $\{\phi_{k}^j\relx\}_{j=0}^{L_w-1}$\edit{, i.e. $\xi\relx$,} at different time instants $k$ differ. Hence we distinguish them by the subscript $(\cdot)_{k}$. 

\edit{We remark that the satisfaction of Assumption~\ref{ass:finitenoise}(i) requires an appropriate choice of the polynomial basis. As mentioned above, the distributions listed in Table~\ref{tab:askey_scheme} admit exact and finite PCEs with only two terms in the chosen bases. Moreover, for arbitrary $\mcl L^2$ random variables one can construct appropriate orthogonal basis functions by the Gram-Schmidt process~\cite{Witteveen2006}. Conceptually, for  $Z\in \spl$,  a straightforward choice is the  orthogonal basis with $\phi^0=1$ and $\phi^1=Z-\mean[Z]$ which implies the exact and finite PCE $\pce{z}^0=\mean[Z]$ and $\pce{z}^1=1$. }

Furthermore, assuming finite PCEs for $X_\text{ini}$ and \edit{for} all i.i.d. $W_{k}$, $k\in \N$ we recall a result which appeared previously in Pan et al.~\cite{Pan21s} and which is based on the work of Mühlpfordt et al.~\cite{muehlpfordt18comments} 
\edit{\begin{lemma}[Exact uncertainty propagation via PCE]\label{lem:no_truncation_error}
		Consider the stochastic LTI system \eqref{eq:RVdynamics} and let Assumption \ref{ass:finitenoise} (i) hold. For a finite prediction horizon $N \in \N^+$, suppose  the inputs $U_{k}$, $k\in \I_{[0,N-1]}$  admit exact finite-dimensional PCEs with at most $L$ terms in basis  $\{\phi^j\}_{j=0}^{L-1}$ , where $L$ is given by
		\begin{subequations}\label{eq:OCPbase}
		\begin{equation}\label{eq:terms}
			L = L_{x} + N(L_w-1) \in \N^+,\vspace*{-0.2cm}
		\end{equation}
and	the finite-dimensional basis $\{\phi^j\}_{j=0}^{L-1}$ reads
		\begin{equation}\label{eq:bases}
			\{\phi^j\}_{j=0}^{L-1} = \{1, \{\phi_x^j\}_{j=1}^{L_x-1}, \{\phi_0^j\}_{j=1}^{L_w-1}, \{\phi_1^j \}_{j=1}^{L_w-1}\dots, \{\phi_{N-1}^j \}_{j=1}^{L_w-1}\},
		\end{equation}
		with $1=\phi_x^0 \relx=\phi_{k}^0(\omega)$ for $k \in \I_{[0,N-1]}$ and $\omega \in \Omega$. 
			\end{subequations}
		Then, the states $X_{k}$, $k\in \I_{[0,N]}$ also admit exact finite-dimensional PCEs with at most $L$ terms in the basis~\eqref{eq:OCPbase}.
\end{lemma}}

The previous lemma implies that as the prediction $N$ grows, the required PCE dimension in \eqref{eq:terms} increases linearly. The reason is that the realizations of $W_k$ are independent at each time instant $k \leq N-1$. Yet, in the context of predictive control, we usually consider a fixed and finite prediction horizon $N$. Hence the finite-dimensional polynomial basis \eqref{eq:bases} enables exact propagation of the uncertainties. Naturally, the exactness of the propagation is meant in the usual $\mcl L^2$ equivalence sense.
To this end, we first replace all random variables of \eqref{eq:RVdynamics} with their PCE coefficients with respect to the considered basis. Then we perform a Galerkin projection onto the basis functions $\phi^j$, $j \in \I_{[0,L-1]}$. With given initial conditions $\pcecoe{x}{j}\ini$ for $j \in \I_{[0,L-1]}$ this yields
\begin{align}\label{eq:PCEcoesDynamics}
\pcecoe{x}{j}\inst{k+1} = A\pcecoe{x}{j}\inst{k} +B\pcecoe{u}{j}\inst{k}+ \pcecoe{w}{j}\inst{k},\quad
\pcecoe{x}{j}\inst{0} = \pcecoe{x}{j}\ini,
\quad  j \in \I_{[0,L-1]}, \quad k \in \I_{[0,N-1]}.
\end{align}

\begin{remarkmod}[Expression of filtered stochastic processes with PCE]
With the polynomial basis given in \eqref{eq:bases}, the causality/non-antipacitivity of the filtration \eqref{eq:filtration} implies that the PCE coefficients of the inputs satisfy 
\begin{subequations}\label{eq:causality_PCE}
\begin{equation} \label{eq:causality_input}
	\pce{u}_k^{j} = 0, \forall j\in \I_{[L_{x}+k(L_w-1),L-1]},\, \forall k \in \I_{[0,N-1]},
\end{equation} 
 We remark that the causality of $X$ trivially holds when \eqref{eq:causality_input} is imposed on systems dynamics \eqref{eq:PCEcoesDynamics}, that is 
\begin{equation}\label{eq:causality_state}
	\pce{x}_k^{j} = 0, \forall j\in \I_{[L_{x}+k(L_w-1),L-1]},\, \forall k \in \I_{[0,N-1]},
\end{equation} 
\end{subequations} see also the discussion by Ou et al.~\cite{Ou21} 
\End
\end{remarkmod}

We conclude the introduction of uncertainty propagation via PCE with a simple example illustrating how Lemma~\ref{lem:no_truncation_error}, the PCE coefficients dynamics \eqref{eq:PCEcoesDynamics}, and the causality condition \eqref{eq:causality_PCE} work.
\begin{example}[PCE coefficient dynamics]\label{exampl:PCEdynamic}
Consider the system $X_{k+1} = X_{k} + U_k +W_k$. Suppose at each time step $k$, we consider the stochastic input $U_k = -0.5 X_k$ to satisfy the filtration \eqref{eq:filtration} and $W_k = 0+1\cdot \phi_{k}$, then for the initial condition $X_{\text{ini}}=X_0= 1 + \phi_x$ the PCEs of all predicted random variables for state, input, and disturbance read
\begin{align*}
U_0&=-0.5 -0.5\phi_x  & W_0&= \phi_{0} &     X_1 &= 0.5 + 0.5\phi_x + \phi_{0}     \\
U_1&=-0.25 -0.25\phi_x-0.5 \phi_{0}  & W_1&= \phi_{1} &  X_2& = 0.25 + 0.25 \phi_x + 0.5 \phi_{0} + \phi_{1} \\
&\vdots& &\vdots&&\vdots&\\
U_{N-1}&=-0.5 X_{N-1}  & W_{N-1}&= \phi_{N-1} & X_N &= 0.5^N + 0.5^N \phi_x + \sum_{k=0}^{N-1} (0.5)^{N-k-1} \phi_{k}.  
\end{align*}
Observe that the random variables \edit{$U_{[0,N-1]}$, $W_{[0,N-1]}$, and $X_{[0,N]}$}
 can be represented in the finite-dimensional polynomial basis $\{\phi^j\}_{j=0}^{L-1}=\{1, \phi_x, \phi_{0},\dots,\phi_{N-1}\}$ with $L_x=L_w=2$ and $L  =  N+2$, and the PCE coefficients for all basis functions $\phi^j$, $j \in \I_{[0,L-1]}$, follow the dynamics \eqref{eq:PCEcoesDynamics}. Note that the causality condition \eqref{eq:causality_PCE} is also satisfied. Finally, we remark that, while in principle one could choose any arbitrary basis for $U_k$---which then implies that the corresponding direction have to be included in $\{\phi^j\}_{j=0}^{L-1}$---, the linear nature of~\eqref{eq:RVdynamics} suggests to express $U_k$ in the basis induced by $W_k$ and $X_{\text{ini}}$.  
\End
\end{example}

%\TF{Please switch off the hyperref links}

\subsection{Data-driven representation of stochastic LTI systems}
The previous subsections have recalled model-based stochastic MPC and the basics of polynomial chaos. However, as the goal of this paper is the analysis of data-driven stochastic MPC, we now turn towards data-driven representations of stochastic LTI systems as proposed in our previous paper~\cite{Pan21s}.

Therein, we have shown that the state and input trajectories---either for random variables \eqref{eq:RVdynamics}, PCE coefficients \eqref{eq:PCEcoesDynamics}, or the realizations \eqref{eq:RelizationDynamics}---can be represented in non-parametric fashion using recorded state, input, disturbance realizations trajectories.\footnote{The results of Pan et al.~\cite{Pan21s} are formulated in the input-output setting and hence they are readily translated to the case of state feedback considered here. } This insight, which is formalized as a stochastic variant of the fundamental lemma~\cite{Pan21s}, extends the classical result of Willems et al.\cite{Willems2005} The underlying key observation is that the dynamics of random variables \eqref{eq:RVdynamics}, PCE coefficients \eqref{eq:PCEcoesDynamics}, and realizations \eqref{eq:RelizationDynamics} share the same system matrices $(A,[B ~I_{\dimx}])$. 

At this point, we first assume that the past \edit{disturbance} realizations are available, i.e., they are a posteriori measurable. We comment on the estimation of  \edit{unmeasured disturbance} realizations in Section~\ref{remark:noiseEst} below.
 
\begin{definition}[Persistency of excitation \cite{Willems2005}] Let $T, N \in \set{N}^+$. A sequence of real-valued inputs $\trar{u}{T-1}$ is said to be persistently exciting of order $N$ if the Hankel matrix
\begin{equation*}
\Hankel_N(\trar{u}{T-1}) \doteq \begin{bmatrix}
u\inst 0  & u\inst 1 &\cdots& u\inst{T-N} \\
u\inst 1  & u\inst 2 &\cdots& u\inst{T-N+1} \\
\vdots& \vdots & \ddots & \vdots \\
u\inst{N-1}&u\inst{N}& \cdots  & u\inst{T-1} \\
\end{bmatrix}
\end{equation*}
is of full row rank.  \End
\end{definition}

Consider the stochastic LTI system \eqref{eq:RVdynamics}, the PCE-coefficient dynamics \eqref{eq:PCEcoesDynamics}, and the realization dynamics \eqref{eq:RelizationDynamics}. The corresponding trajectory tuples are $\trar{(X,U,W)}{T-1}$, $\trar{\pce{(x,u,w)}}{T-1}^j, \, j \in \I_{[0,L]}$, and $\trar{(x,u,w)}{T-1}$ respectively. 
\begin{lemma}[Fundamental lemma for stochastic LTI systems]\label{lem:RVfundamental} 
Let Assumption~\ref{ass:AB} and Assumption~\ref{ass:finitenoise} (i) hold. Consider the stochastic LTI system \eqref{eq:RVdynamics}, its corresponding dynamics of PCE coefficients~\eqref{eq:PCEcoesDynamics}, and the realizations~\eqref{eq:RelizationDynamics}. Given a $T$-length realization trajectory tuple $\trar{(x,u,w)}{T-1}$ of \eqref{eq:RelizationDynamics}.
Let $\trar{(u,w)}{T-1}$ be persistently exciting of order $\dimx +N$. Then the following statements hold:
\begin{subequations}
\begin{itemize}
\item[(i)] $\trar{(\tilde{x},\tilde{u},\tilde{w})}{N-1}$ is  a state-input-\edit{disturbance} realization trajectory tuple  of \eqref{eq:RelizationDynamics} of length $N$ if and only if there exists a $g\in \R^{T-N+1} $ such that
\begin{equation} \label{eq:realization_funda}
\Hankel_N(\trar{z}{T-1})g= \trar{\tilde{z}}{N-1}
\end{equation} 
holds for all $(\mbf{z}, \tilde{\mbf{z}})\in \{(\mbf{x},\tilde{\mbf{x}}), (\mbf{u},\tilde{\mbf{u}}), (\mbf{w},\tilde{\mbf{w}})\} $.
\item[(ii)]
For all $j \in \I_{[0,L-1]}$, $\trar{ (\tilde{\pce x}, \tilde{\pce u}, \tilde{\pce w})}{N-1}^j $ is  a state-input-\edit{disturbance}  PCE coefficient trajectory tuple of \eqref{eq:PCEcoesDynamics} of length $N$ if and only if there exists a $\pcecoe{g}{j}\in \R^{T-N+1}$ such that 
\begin{equation} \label{eq:mixed_funda}
\Hankel_N(\trar{z}{T-1}) \pcecoe{g}{j}= \trar{ \tilde{\pce z}}{N-1}^j
\end{equation} holds for all $(\mbf{z}, \tilde{\pce{z}})\in \{(\mbf{x},\tilde{\pce{x}}), (\mbf{u},\tilde{\pce{u}}), (\mbf{w},\tilde{\pce{w}})\} $ and  all $j \in \I_{[0,L-1]}$.
\item[(iii)]
$ (\tilde{\mbf{X}}, \tilde{\mbf{U}}, \tilde{\mbf{W}})_{[0,N-1]}$ is  a state-input-\edit{disturbance}  random variable trajectory tuple of \eqref{eq:RVdynamics} of length $N$ if and only if there exists $G \in \splx{T-N+1} $ such that 
\begin{equation} \label{eq:RVfunda}
\Hankel_N(\trar{z}{T-1}) G=\tilde{\mbf{Z}}_{[0,N-1]}
\end{equation} 
holds for all $(\mbf{z}, \tilde{\mbf{Z}})\in \{(\mbf{x},\tilde{\mbf{X}}), (\mbf{u},\tilde{\mbf{U}}), (\mbf{w}, \tilde{\mbf{W}})\} $.\End
\end{itemize}
\end{subequations}
\end{lemma}
Assertion (i) follows from the proof of the deterministic fundamental lemma Theorem 1 by Willems et al.~\cite{Willems2005} The proof is originally given in the behavioral language, while De Persis et al. \cite{DePersis19} use state-space concepts. Assertion (ii)  is presented in Corollary 2 of Pan et al. \cite{Pan21s} in the input-output setting and Assertion (iii) follows from Lemma 4 of the same paper.

We see that Lemma \ref{lem:RVfundamental} cannot be applied to reformulate OCP \eqref{eq:stochasticOCP} directly, since in \eqref{eq:stochasticOCP} we predict a state trajectory of length $N+1$ while the input-\edit{disturbance} trajectories are of length $N$. The reason is the $N$-th input and disturbance are not considered neither in objective nor in the constraints. 
Therefore, we extend the statement (iii) of Lemma \ref{lem:RVfundamental} as follows:

\begin{lemma}[Hankel matrix representation for state-feedback MPC]\label{pro:preXN} Let Assumption~\ref{ass:AB} and \ref{ass:finitenoise} (i) hold. Consider a input-disturbance realization tuple $\trar{(u,w)}{T-1}$ and its corresponding  state  trajectory $\trar{x}{T}$ of \eqref{eq:RelizationDynamics}. Let $\trar{(u,w)}{T-1}$ be persistently exciting of order $\dimx +N+1$. Then $\tilde{\mbf{X}}_{[0,N]}$ is a corresponding state trajectory of $(\tilde{\mbf{U}},\tilde{\mbf{W}})_{[0,N-1]}$ w.r.t. \eqref{eq:RVdynamics} if and only if there exists $G \in \splx{T-N+1} $ such that
\begin{equation}\label{eq:preXN}
\begin{bmatrix}
\Hankel_{N+1}(\trar{x}{T})\\
\Hankel_{N}(\trar{u}{T-1})\\
\Hankel_{N}(\trar{w}{T-1})\\
\end{bmatrix} G = \begin{bmatrix}
\tilde{\mbf{X}}_{[0,N]}\\
\tilde{\mbf{U}}_{[0,N-1]}\\
\tilde{\mbf{W}}_{[0,N-1]}\\
\end{bmatrix}.
\end{equation}
\end{lemma}
\begin{proof}
Consider appending $\trar{(u,w)}{T-1}$ with arbitrary $u_T$ and $w_T$, we see the new realization trajectory $\trar{(u,w)}{T}$ is at least persistently exciting of order $\dimx +N+1$ by assumption. Then, by Statement (iii) of Lemma \ref{lem:RVfundamental}, there exists $G \in \splx{T-N+1} $ such that
\begin{equation}\label{eq:HankelN}
\begin{bmatrix}
\Hankel_{N+1}(\trar{x}{T})\\
\Hankel_{N+1}(\trar{u}{T})\\
\Hankel_{N+1}(\trar{w}{T})
\end{bmatrix} G = \begin{bmatrix}
\tilde{\mbf{X}}_{[0,N]}\\
\tilde{\mbf{U}}_{[0,N]}\\
\tilde{\mbf{W}}_{[0,N]}
\end{bmatrix}.
\end{equation}
By separating the rows of the Hankel matrix in \eqref{eq:HankelN} that are corresponding to $\tilde{\mbf{U}}_{N}$ and $\tilde{\mbf{W}}_{N}$ 
\begin{equation*}
   \begin{bmatrix}
\Hankel_{N+1}(\trar{x}{T})\\
\Hankel_{N}(\trar{u}{T-1})\\
\hline
\Hankel_{1}(\mbf{u}_{[N,T]})\\
\Hankel_{N}(\trar{w}{T-1})\\
\hline
\Hankel_{1}(\mbf{w}_{[N,T]})\\
\end{bmatrix} G = \begin{bmatrix}
\tilde{\mbf{X}}_{[0,N]}\\
\tilde{\mbf{U}}_{[0,N-1]}\\
\hline
\tilde{\mbf{U}}_{N}\\
\tilde{\mbf{W}}_{[0,N-1]}\\
\hline
\tilde{\mbf{W}}_{N}
\end{bmatrix}.
\end{equation*}
We conclude with \eqref{eq:preXN} by removing rows $\Hankel_{1}(\mbf{u}_{[N,T]})$ and $\Hankel_{1}(\mbf{w}_{[N,T]})$ to avoid representation of $\tilde{\mbf{U}}_N$ and $\tilde{\mbf{W}}_N$.
\end{proof}
The statements (i) and (ii) in Lemma \ref{lem:RVfundamental} can also be extended similarly.
We conclude this section with a brief remark on how past \edit{disturbance} realizations can be estimated. 

\begin{remarkmod}[Estimation of past process \edit{disturbance} realizations]\label{remark:noiseEst}
The current discussion is based on the knowledge of \edit{the past disturbance realizations}. It is useful for the case \edit{when the disturbances are measured a posteriori.} For unmeasured disturbances, one may use the recorded state-input trajectories $\mbf{x}_{[0,T]}$ and $\mbf{u}_{[0,T-1]}$ to approximate the past disturbance realizations by
\begin{itemize}
\item[(i)] the maximum likelihood estimate
\begin{subequations}\label{eq:minerror}
\begin{align}
  \hat{\mbf{w}}_{[0,T-1]}&= \argmin_{\mbf{w}_{[0,T-1]}} -\sum_{k=0}^{T-1} \log p_W(w \inst{k}), \\ \label{eq:leftkern}
   \quad \text{subject to }   \big(\xxdp-\wwd \big)&\bigg( I_T - 
\begin{bmatrix}
\xxd \\
\uud
\end{bmatrix}^\rinv 
\begin{bmatrix}
\xxd \\
\uud
\end{bmatrix}\bigg)=0,
\end{align}
where $p_W$ denotes the probability density function of $W_k$,
\end{subequations}
\item[(ii)]or the least-squares estimate
\begin{align*}%\label{eq:minerrorGD}
 \hat{\mbf{w}}_{[0,T-1]} = \argmin_{\mbf{w}_{[0,T-1]} } { \mbf{w}^{\top}_{[0,T-1]} \mbf{w}_{[0,T-1]} } , \quad \text{subject to }   \eqref{eq:leftkern}.
\end{align*}
with
\begin{equation}\label{eq:estimdateWd}
    \Hankel_1(\hat{\mbf{w}}_{[0,T-1]}) = \xxdp\bigg( I_T-\begin{bmatrix}
\xxd \\
\uud
\end{bmatrix}^\rinv 
\begin{bmatrix}
\xxd \\
\uud
\end{bmatrix}\bigg) .
\end{equation} 
\end{itemize}
We remark that for Gaussian \edit{disturbances} \eqref{eq:minerror} and \eqref{eq:estimdateWd} are equivalent. However, for uniformly distributed \edit{disturbances}---since the maximum likelihood estimate \eqref{eq:minerror} does not admit a unique solution due to the constant probability density function---one rather considers \eqref{eq:estimdateWd}. \edit{For further details on the disturbance estimation, see~\cite{Pan21s}. } 
\edit{Moreover, one can show that any estimated disturbance sequence satisfying \eqref{eq:leftkern} implicitly determines an LTI system with $(\hat{A},\hat{B})$~\cite{Pan21s}. This means that the Hankel matrix representations, i.e. the results of Lemmas~\ref{lem:RVfundamental} and~\ref{pro:preXN}, also hold for the the state-input-disturbance trajectories with estimated disturbances.}
\End
\end{remarkmod}

\section{Data-driven Stochastic Predictive Control}\label{sec:OCPS}
Now, we are ready to turn towards the main topic of the paper, i.e. data-driven stochastic MPC. First, we present the data-driven counterpart of OCP~\eqref{eq:stochasticOCP}. Then, we propose a data-driven predictive control scheme with the analysis of its closed-loop properties.

\subsection{Data-driven stochastic OCP via PCE}\label{sec:dataOCP}
\edit{In the following, consider $	\pcecoe{v}{[0,L-1]}\doteq [\pcecoe{v}{0\top},\pcecoe{v}{1\top},\dots,\pcecoe{v}{L-1\top}]^\top$ as the vectorization of PCE coefficients over PCE dimensions.}
Let Assumption~\ref{ass:finitenoise} (i) hold. Consider the availability of realization trajectories $\mbf{x}_{[0,T]}$, $\trar{u}{T-1}$, and $\trar{w}{T-1}$ of \eqref{eq:RVdynamics}. Suppose that $\trar{(u,w)}{T-1}$ is persistently exciting of order $N+\dimx +1$. As before, the PCEs with $L \in \N^+$ dimensional polynomial basis are determined by \eqref{eq:OCPbase} and PCE coefficients $\bar{\pce{x}}^{[0,L-1]}\inst k$ and $\pcecoe{w}{[0,L-1]} \inst{i}$, $i \in \I_{[k,k+N-1]}$ are given. Then, the data-driven reformulation of OCP~\eqref{eq:stochasticOCP} using PCE reads
\begin{subequations}\label{eq:H_PCE_SOCP}
\begin{align}
\edit{  V_{N}(\bar{\pce{x}}^{[0,L-1]}_k)
 \doteq} \min_{
\substack{
 \pce{x}^{[0,L-1]}_{[k,k+N]|k} \in \R^{\dimx(N+1)L}, \\ 
 \pce{u}^{[0,L-1]}_{[k,k+N-1]|k} \in \R^{\dimu N L}, \\
  \pcecoe{g}{[0,L-1]} \in \R^{(T-N+1)L} 
  }
 }
  \sum_{i=k}^{k+N-1} \sum_{j=0}^{L-1} \Big(&\| \pcecoe{x}{j} \pred{i}{k}\|^2_Q  +\|\pcecoe{u}{j}\pred{i}{k}\|^2_R\Big)\langle \phi^j,\phi^j\rangle + \sum_{j=0}^{L-1}  \|\pcecoe{x}{j}\pred{k+N}{k}\|^2_P\langle \phi^j,\phi^j\rangle \label{eq:H_PCE_SOCP_obj}    \\   \text{subject to } \quad
\begin{bmatrix}
 \text{ \multirow{2}{*}{$\Hankel_{N+1}(\mbf{x}_{[0,T]})$}}\\
 \\
\Hankel_{N}(\mbf{u}_{[0,T-1]})\\
\Hankel_{N}(\mbf{w}_{[0,T-1]})
\end{bmatrix}    
      \pcecoe{g}{j}&= 
      \begin{bmatrix}
\bar{\pce{x}}^{j}\inst{k}\\      
          \pce{x}^j_{[k+1,k+N]|k}   \\
          \pce{u}^j_{[k,k+N-1]|k}\\
          \pcecoe{w}{j}\inst{[k,k+N-1]} 
      \end{bmatrix}, \quad \forall j\in \set{I}_{[0,L-1]}
    \label{eq:H_PCE_SOCP_hankel}   
 \\
    \pcecoe{x}{0}\pred{i}{k}  \pm \sigma(\varepsilon_x)\sqrt{\sum_{j=1}^{L-1} {(\pcecoe{x}{j}\pred{i}{k})}^2\langle\phi^j,\phi^j \rangle} \in  \mathbb X,
   & \quad  
    \pcecoe{u}{0}\pred{i}{k}  \pm \sigma(\varepsilon_u)\sqrt{\sum_{j=1}^{L-1} {(\pcecoe{u}{j}\pred{i}{k}})^2\langle \phi^j,\phi^j \rangle} \in \mathbb U,  \quad \forall i \in \set{I}_{[k,k+N-1]} \label{eq:PCE_chance}\\
 \pcecoe{x}{0}\pred{k+N}{k}  \in  \Xf, & \quad \sum_{j=1}^{L-1} \pcecoe{x}{j}\pred{k+N}{k}\pcecoe{x}{j\top}\pred{k+N}{k} \langle \phi^j,\phi^j\rangle\leq \Gamma,   \label{eq:PCE_Xf}\\ 
     \pce{u}\pred{i}{k}^{j} = 0,& \quad~ \forall j\in \I_{[L_{x}+(i-k)(L_w-1),L-1]} \label{eq:causality},
    \end{align} 
 where $\pce{x}^j_{[k+1,k+N]|k} \doteq (\pcecoe{x}{j\top}_{k+1|k},\pcecoe{x}{j\top}_{k+2|k},\cdots, \pcecoe{x}{j\top}_{k+N|k})^\top$ and $\pce{u}^j_{[k,k+N-1]|k} \doteq (\pcecoe{u}{j\top}_{k|k},\pcecoe{u}{j\top}_{k+1|k},\cdots, \pcecoe{u}{j\top}_{k+N-1|k})^\top$.
\end{subequations} \vspace*{1mm}

As shown in 
\eqref{eq:H_PCE_SOCP_hankel}
 OCP \eqref{eq:H_PCE_SOCP}, we  utilize \eqref{eq:mixed_funda} in Lemma \ref{lem:RVfundamental} to describe future PCE coefficients  for states and inputs with the recorded \textit{realization data} $\mbf{x}_{[0,T]}$, $\trar{u}{T-1}$,$\trar{w}{T-1}$, and the knowledge of future PCE coefficients of the \edit{disturbances}. 
Observe that the Hankel matrix $\Hankel_{N+1}(\trar{x}{T})$ is used to predict  PCE coefficients of the states from $k$ to $k+N$, and thus $\trar{(u,w)}{T-1}$ is required to be persistently exciting of order $\dimx +N+1$, see Lemma~\ref{pro:preXN}.

The chance constraint reformulation from \eqref{eq:chance} to \eqref{eq:PCE_chance} is standard and follows Farina et al.~\cite{Farina13p} Especially, for Gaussian random variables \eqref{eq:chance} and \eqref{eq:PCE_chance} are equivalent when $\sigma$ is chosen according to the Gaussian cumulative distribution function, cf. Theorem 2.1 of Calafiore et al.~\cite{Calafiore06} Moreover, a conservative choice for all distributions is $\sigma(\varepsilon_z) = \sqrt{(2-\varepsilon_z)/\varepsilon_z}$, $z\in\{x,u\}$. The transformation of the terminal constraints \eqref{eq:terminalCons} to \eqref{eq:PCE_Xf} follows from the properties of finite PCE in \eqref{eq:PCEmoments}. In addition, the causality condition \eqref{eq:causality_input} is considered explicitly in \eqref{eq:causality}.

The following result summarises the conditions for the equivalence of the model-based OCP~\eqref{eq:stochasticOCP} and its data-driven counterpart~\eqref{eq:H_PCE_SOCP}. It is a direct consequence of Theorem 1 of Pan et al.~\cite{Pan21s}

\begin{corollary}[Equivalence of stochastic OCPs]\label{col:equivOCP}  Consider the stochastic OCP \eqref{eq:stochasticOCP} with random variables in a filtered probability space fulfilling \eqref{eq:filtration} and the data-driven reformulation \eqref{eq:H_PCE_SOCP} of OCP \eqref{eq:stochasticOCP}. Let Assumption \ref{ass:AB} and \ref{ass:finitenoise} hold. Suppose that the $L \in \N^+$ dimensional polynomial basis is determined by \eqref{eq:OCPbase} and realization trajectories $\mbf{x}_{[0,T]}$, $\trar{u}{T-1}$ and $\trar{w}{T-1}$, with $\trar{(u,w)}{T-1}$ persistently exciting of order $N+1+\dimx$, are known. 

Moreover, if  \edit{ the chance constraints are box constraints, i.e. $\mbb{Z}=[\underline{z},\overline{z}]$ for $\mbb{Z} \in \{\mbb{X},\mbb{U}\}$ and} the reformulation of chance constraints \eqref{eq:chance} to \eqref{eq:PCE_chance} is exact, then \eqref{eq:stochasticOCP} $\Leftrightarrow$ \eqref{eq:H_PCE_SOCP}, i.e., for any optimal solution $ X^\star\pred{i+1}{k}, U^\star\pred{i}{k}$ for $i\in \set{I}_{[k,k+N-1]}$ to \eqref{eq:stochasticOCP} there exists a $\pcecoe{g}{[0,L-1],\star} \in  \R^{(T-N+1)\cdot L} $ such that  $\pce{x}^{[0,L-1],\star}\pred{i+1}{k},\,
 \pce{u}^{[0,L-1],\star}\pred{i}{k},\,
  \pcecoe{g}{[0,L-1],\star}$, $i \in \I_{[k,k+N-1]}$
is an optimal solution to \eqref{eq:H_PCE_SOCP} \edit{and vice versa}. In particular, it holds that $X^\star\pred{i+1}{k} = \sum_{j=0}^{L-1} \pcecoe{x}{j,\star}_{i+1|k} \phi^j, U^\star\pred{i}{k}  = \sum_{j=0}^{L-1} \pcecoe{u}{j,\star}_{i|k} \phi^j$, $i \in \I_{[k,k+N-1]}$ is an optimal solution to \eqref{eq:stochasticOCP}. \End
\end{corollary}
\edit{
\begin{remarkmod}[Comparison to model-based formulation] We remark OCP~\eqref{eq:H_PCE_SOCP} is also equivalent to its model-based counterpart by replacing~\eqref{eq:H_PCE_SOCP_hankel} with the PCE coefficient dynamics~\eqref{eq:PCEcoesDynamics}, see Theorem~1 of Pan et al.~\cite{Pan21s} Hence, if one considers nominal settings (i.e. data-driven stochastic MPC with exact disturbance measurements is pitted against model-based stochastic MPC with exact models), the model-based scheme gives the very same solution while being computationally more efficient. However, if one considers that the model is uncertain, e.g.
	\[
X^+ = A(\Theta) X + B(\Theta) U + W
	\]
	and the parameter $\Theta$ is modelled, e.g., via PCE (such that its realization stays constant over the entire horizon), cf.~\cite{Wan2022}, then the proposed scheme has the clear advantage of being conceptually much simpler.  In \edit{our} setting, the proposed approach reduces the burden of combining PCE for $\Theta$ with the ones for $X$, $U$, $W$ by instead resorting to the data-driven fundamental lemma which implies to use PCE for $X$, $U$, and $W$. We refer to more detailed discussions by Faulwasser et al.~\cite{tudo:faulwasser22f} \End
\end{remarkmod}
}

\begin{remarkmod}[Data-driven design of terminal ingredients]\label{remark:dataTerminal}
In the model-based OCP~\eqref{eq:stochasticOCP}  the terminal ingredients $K$, $P$, and $\Gamma$ are constructed using knowledge of $A$ and $B$, see \eqref{eq:modelbasedP} and \eqref{eq:modelbasedgamma}. Next, we utilize the recorded states $\trar{x}{T}$ and inputs $\trar{u}{T-1}$ combined with the estimated disturbance trajectory $\trar{w}{T-1}$ to compute the terminal ingredients without knowledge of $A$ and $B$.
 Suppose $\mbf{(u,w)}_{[0,T-1]}$ is persistently exciting of order more than $\dimx +1$, then there exists a matrix $H \in \R^{T \times \dimx}$ that relates the data $\mbf{x}_{[0,T]}$, $\mbf{u}_{[0,T-1]}$, and $\mbf{w}_{[0,T-1]}$ to a given feedback $K \in \R^{\dimu\times \dimx}$ and to  $(A,B)$ with
\begin{align}\label{eq:IK} 
  \begin{bmatrix}
\xxd \\
\uud
\end{bmatrix}H = \begin{bmatrix}
I_{\dimx} \\
K
\end{bmatrix}\quad& \text{and } \quad 
 \left(\xxdp - \wwd \right)H =  A+BK .   
\end{align}
Equation \eqref{eq:IK} has been proposed by Dörfler et al.~\cite{Dorfler21c} as an extension to Theorem 2 of De Perisis et al.~\cite{DePersis19} including additive process disturbances.

In the following, we  use 
\begin{equation*}
M H \doteq \left(\xxdp - \wwd \right) H
\end{equation*}
as the data-driven counterpart of $A+BK$. Then, we have the following data-driven problem to obtain $H$ and $K$ 
\begin{subequations}
\begin{align*}
\min_{
\substack{
\tilde{P} \in \R^{\dimx\times\dimx}, K\in \R^{\dimx \times \dimu},\\
H \in \R^{T\times \dimx}}
} &\trace\left(Q\tilde{P}+K^\top R K  \tilde{P}\right)\\
 \text{ subject to }\quad \eqref{eq:IK},\quad & \tilde{P}  \succeq I_{\dimx}, \quad
M H \tilde{P} H^\top  M^\top - \tilde{P}  +I_{\dimx} \preceq 0. 
\end{align*}
\end{subequations}
 Its convex reformulation reads \cite{Dorfler21c}
	\begin{subequations}\label{eq:dataPKG}
		\begin{align}
			\min_{
				\substack{
					X_1 \in \R^{\dimu\times\dimu}, X_2 \in \R^{T\times\dimx}
				}
			}\trace\left(Q \xxd X_2\right)&+\trace\left(X_1\right) \\
			\text{subject to }\quad
			\begin{bmatrix}
				\xxd X_2 - I &  M X_2   \\
				\star  & 	\xxd X_2
			\end{bmatrix}\succeq 0, \quad&
			\begin{bmatrix}
				X_1 &  R^{\frac{1}{2}} \uud X_2 \\
				\star & 
				\xxd X_2
			\end{bmatrix}\succeq  0,
		\end{align}
	\end{subequations}
	with $K = \uud X_2 \left(\xxd X_2\right)^{-1}$ and $H = X_2 \left(\xxd X_2\right)^{-1} $.

Similarly, we obtain the data-driven formulation of the Lyapunov equation  \eqref{eq:modelbasedgamma}  to compute $\Gamma$ 
\begin{subequations}\label{eq:gamma}
\begin{align}
M H \Gamma H^\top M^\top - \Gamma +\Ucov &=0.
\end{align}
Further, with $K$ and $H$ we have the following discrete-time Lyapunov equation to determine $P$ 
\begin{align}
P - H^\top M^\top  P M H- K^\top R K &= Q. \label{eq:dataRiccati} 
\end{align}
\end{subequations}
 \End 
\end{remarkmod}

\begin{remarkmod}[Numerical implementation with regularization]\label{rem:slack}
Observe that any small inaccuracy or disturbance in the initial condition caused by numerical error might jeopardize the feasibility of OCP \eqref{eq:H_PCE_SOCP}. To overcome this issue, we add a slack variable $c \in \R^{L}$ in  \eqref{eq:H_PCE_SOCP_hankel}, i.e.,
\[
\begin{bmatrix}
	\text{ \multirow{2}{*}{$\Hankel_{N+1}(\mbf{x}_{[0,T]})$}}\\
	\\
	\Hankel_{N}(\mbf{u}_{[0,T-1]})\\
	\Hankel_{N}(\mbf{w}_{[0,T-1]})
\end{bmatrix}    
\pcecoe{g}{j} = 
\begin{bmatrix}
	\bar{\pce{x}}^{j}\inst{k} + c^j\\      
	\pce{x}^j_{[k+1,k+N]|k}   \\
	\pce{u}^j_{[k,k+N-1]|k}\\
	\pcecoe{w}{j}\inst{[k,k+N-1]} 
\end{bmatrix}, \quad \forall j\in \set{I}_{[0,L-1]},
\]
where $c^j$ denotes the $j$-th element of $c$ in this equation. Consequently, the penalty term $\beta \|c\|_1$ is included in the objective function with $\beta\gg 0$. The use of slack variables is widely considered in deterministic data-driven predictive control, e.g. one-norm penalization~\cite{Coulson2019} and two-norm penalization~\cite{Berberich20}.
\End
\end{remarkmod}

\edit{\begin{remarkmod}[Numerical implementation with multiple-shooting] The basis construction for exact PCEs in Lemma~\ref{lem:no_truncation_error} indicates that the dimension of the basis~\eqref{eq:OCPbase} grows linearly as $N$ increases. Moreover, the number of the rows of the Hankel matrix in the equality constraint \eqref{eq:H_PCE_SOCP_hankel} of OCP~\eqref{eq:H_PCE_SOCP} is $(2\dimx+\dimu)N+\dimx$. Hence, the number of the decision variables in OCP~\eqref{eq:H_PCE_SOCP} grows quadratically with the prediction horizon $N$. Ou et al.~\cite{Ou23} consider a segmented prediction horizon composed of shorter intervals.
 Furthermore, the consecutive solution pieces are coupled with continuity constraints only  matching the first two moments, which leads to a substantial reduction of  the dimension of PCE basis and of the number of decision variables. We remark that the idea of segmenting the whole predicition horizon resembles the classic concept of multiple shooting~\cite{Bock84a}; we refer to O'Dwyer et al.~\cite{o2021data} for the first implementation of multiple shooting in data-driven control. \End
\end{remarkmod}
}

\subsection{ Data-driven stochastic predictive control algorithm}\label{sec:MPCs}
Now, we propose a data-driven stochastic predictive control scheme based on \eqref{eq:H_PCE_SOCP}. Importantly, we will give recursive feasibility guarantees. The basic idea is to consider an alternative feasibility-guaranteeing initial condition in case \eqref{eq:stochasticOCP} is infeasible with the current state realization as its initial condition. For the underlying conceptual framework in model-based stochastic predict control, we refer to Farina et al.\cite{Farina13p}

With the current state measurement $x_k$, one can set the initial condition $\bar{\pce{x}}_k^{[0,L-1]}$ in OCP~\eqref{eq:H_PCE_SOCP} as
\begin{equation}\label{eq:initialxk}
\bar{\pce{x}}^{0}_{k} = x\inst{k},\, \bar{\pce{x}}^{j'}_{k} =0 ,\quad \forall j' \in \I_{[1,L-1]}.
\end{equation}
However, in this case, for process \edit{disturbances} $W_k$ with infinite support, OCP \eqref{eq:H_PCE_SOCP} can be infeasible, e.g., in case of a very large \edit{disturbance} realization $w_{k-1}$. Thus, we provide a backup choice of the initial condition based on the optimal solution of the previous OCP solved. 

\begin{lemma}[Recursive feasibility]\label{lemma:recurFeasi}
Let Assumptions~\ref{ass:AB}--\ref{ass:finitenoise} hold. Suppose that at time instant $k-1$, OCP~\eqref{eq:H_PCE_SOCP} is constructed with a $L \in \N^+$ dimensional polynomial basis $\{\phi^j\}_{j=0}^{L-1}$ determined by \eqref{eq:OCPbase}. If OCP~\eqref{eq:H_PCE_SOCP} is feasible at time instant $k-1$ with initial condition $\bar{\pce{x}}^{[0,L-1]}_{k-1} \in \R^{\dimx L}$, then OCP~\eqref{eq:H_PCE_SOCP} constructed with $L'=L+L_w-1 $ dimensional polynomial basis $\{\phi^j\}_{j=0}^{L'-1}=\{\{\phi^j\}_{j=0}^{L-1}, \{\phi^j_{k+N-1}\}_{j=1}^{L_w-1}\}$
is feasible with the initial condition
\begin{equation}\label{eq:initialXk}
\bar{\pce{x}}^{j}_{k}= \pcecoe{x}{j,\star}\pred{k}{k-1},  \forall j \in \I_{[0,L-1]}, \quad  \bar{\pce{x}}^{j'}_{k}= 0, \forall j' \in \I_{[L,L'-1]},
\end{equation} 
where $\pcecoe{x}{j,\star}\pred{k}{k-1}$ are the prediction of the PCE coefficients for $X_{k}$ with respect to the optimal solution at time instant $k-1$.\End 
\end{lemma}
\begin{proof}
Similar to the arguments in the proof of recursive feasibility for deterministic MPC, we construct a shifted version of the last optimal solution as a candidate solution for the current OCP. In terms of PCE coefficients we obtain
\begin{subequations}\label{eq:feasibleShift}
\begin{align}
\text{for } j \in \I_{[0,L-1]}:  \{\tilde{\pce{u}}^j_{i|k}\}_{i=k}^{k+N-2} & = \{\pcecoe{u}{j,\star}_{i|k-1}\}_{i=k}^{k+N-2}, \quad \tilde{\pce{u}}^j_{k+N-1|k} = K \pcecoe{x}{j,\star}_{k+N-1|k-1}
 \\
\{\tilde{\pce{x}}^j_{i|k} \}_{i=k}^{k+N-1} &= \{ \pcecoe{x}{j,\star}_{i|k-1}\}_{i=k}^{k+N-1}, \quad \tilde{\pce{x}}^j_{k+N|k} = M H \pcecoe{x}{j,\star}_{k+N-1|k-1}, \label{eq:shiftedX} \\ \vspace*{3mm}
 \text{for } j \in \I_{[L,L']}:  \{\tilde{\pce{u}}^j_{i|k}\}_{i=k}^{k+N-1} &= 0,  \quad \{\tilde{\pce{x}}^j_{i|k} \}_{i=k}^{k+N-1}  = 0,\quad \tilde{\pce{x}}^j_{k+N|k} =\pce{w}^j_{k+N-1 }. \label{eq:shiftedAfterL}
\end{align}
\end{subequations}
Here we remark that $MH$ is the data-driven counterpart of $(A+BK)$ as shown in \eqref{eq:IK}. Observe that by considering $\tilde{X}_{k+N|k}=  M H\tilde{X}_{k+N-1|k-1}+ W_{k+N-1}$ for the shifted trajectory, the basis of $W_{k+N-1}$ should be appended to the original basis. That is, we construct a polynomial basis $\{\phi^j\}_{j=0}^{L'-1}=\{\{\phi^j\}_{j=0}^{L-1}, \{\phi^j_{k+N-1}\}_{j=1}^{L_w-1}\}$ with $L'=L+L_w-1 $. 

To the end of feasibility of \eqref{eq:feasibleShift} in OCP~\eqref{eq:H_PCE_SOCP}, we first set $\tilde{\pce{x}}^j_{k|k}$ as the initial condition $\bar{\pce{x}}^{j}_{k}$, which is equivalent to \eqref{eq:initialXk}. Furthermore, with respect to the terminal constraints \eqref{eq:PCE_Xf}, we see by Assumption~\ref{ass:terminalinX}, $\tilde{\pce{x}}^0_{k+N|k}=  M H \pcecoe{x}{0,\star}_{k+N-1|k-1}\in \Xf$ holds since $\pcecoe{x}{0,\star}_{k+N-1|k-1} \in \Xf$. Moreover, for the constraint on the covariance of the terminal state, we have
\begin{equation*}
\begin{aligned}
\covar\big[\tilde{X}\pred{k+N}{k}\big]&= \sum_{j=1}^{L'-1} \tilde{\pce{x}}^{j}\pred{k+N}{k} \tilde{\pce{x}}^{j\top}\pred{k+N}{k} \langle \phi^j,\phi^j\rangle \\
&=  \sum_{j=1}^{L-1} M H\pce{x}^{j,\star}\pred{k+N-1}{k-1}  \pce{x}^{j,\star\top}\pred{k+N-1}{k-1} H^\top M^\top  \langle \phi^j,\phi^j\rangle  + \sum_{j=1}^{L_w-1} \pce{w}^{j}\pce{w}^{j\top}\langle \phi^j_{k+N-1},\phi^j_{k+N-1}\rangle
   \\
   &=  MH\left(\sum_{j=1}^{L-1} \pce{x}^{j,\star}\pred{k+N-1}{k-1} \pce{x}^{j,\star\top}\pred{k+N-1}{k-1} \langle \phi^j,\phi^j\rangle\right)H^\top M^\top  +  \sum_{j=1}^{L_w-1} \pce{w}^{j} \pce{w}^{j\top}\langle \phi^j_{k+N-1},\phi^j_{k+N-1}\rangle.
 \end{aligned} 
 \end{equation*} 
 
Furthermore, exploiting the terminal covariance constraint \eqref{eq:terminalCons} at time instant $k-1$ and the finite (co-)variance property~ \eqref{eq:finitrnoise} of $W_{k+N-1}$ we have 
   \begin{align*}     
\covar\big[\tilde{X}\pred{k+N}{k}\big]=   MH\covar\big[X^\star\pred{k+N-1}{k-1}\big] H^\top M^\top +\covar\big[W_{k+N-1}\big] \stackrel{\eqref{eq:terminalCons}, \eqref{eq:finitrnoise}}{\leq}    
    M H\Gamma H^\top M^\top  + \Ucov \stackrel{\eqref{eq:modelbasedgamma}}{=}
 \Gamma.    
\end{align*}
Note the last equality holds by the definition of $\Gamma$ in \eqref{eq:modelbasedgamma}. Consequently, we conclude that at time instant $k$, OCP \eqref{eq:H_PCE_SOCP} with polynomial basis  $\{\phi^j\}_{j=0}^{L'-1}$ is feasible with \eqref{eq:initialXk} as its initial condition and \eqref{eq:feasibleShift} is a feasible solution.
\end{proof}

\begin{algorithm}[t!]
        \caption{Stochastic data-driven predictive control algorithm with OCP \eqref{eq:H_PCE_SOCP}} \label{alg:datadrivenSMPC}
\algorithmicrequire $T$, $N$, $L_w$, $q\leftarrow 0$, $\tilde{V}_0 \leftarrow +\infty$, $k \leftarrow 0$\\
\textbf{Data collection and pre-processing:}
\begin{algorithmic}[1]
\State Sample $u_k \in \set{U}$, $k\in I_{[0,T-1]}$ randomly for $\trar{u}{T-1}$   \label{step:pre_1}
\State Apply $\trar{u}{T-1}$ to System \ref{eq:RVdynamics}, record  $\trar{x}{T}$ 
\State Estimate $\trar{w}{T-1}$ by \eqref{eq:minerror} or  \eqref{eq:estimdateWd}
\State If $\trar{(u,w)}{T-1}$ is persistently of exciting of order less than $N+n_x+1$, go to Step \ref{step:pre_1}, else go to  Step  \ref{step:pre_5}
\State Determine $P$, $K$, $G$, $\Gamma$ by \eqref{eq:dataPKG} and \eqref{eq:gamma}  with $\trar{u}{T-1}$, $\trar{x}{T}$ and the estimated $\trar{w}{T-1}$\label{step:pre_5}
\State Construct \eqref{eq:H_PCE_SOCP} with $\mbf{x}_{[0,T]}$, $\trar{u}{T-1}$ and $\trar{w}{T-1}$ and $P$, $\Gamma$ %with $\trar{(x,u,w)}{T-1}$.
\end{algorithmic}
\textbf{Online optimization:}
\begin{algorithmic}[1]
\State  Measure or estimate $x_k$ \label{step:MPC_0}
\State \edit{$V_{N}^{\text{m}} \leftarrow +\infty$, $V_{N}^{\text{b}} \leftarrow +\infty$}
\State Measured initial condition: $L \leftarrow 1 + N (L_w -1)$,$\quad$  $ \bar{\pce{x}}^0_{k}\leftarrow  x_k, \, \bar{\pce{x}}^{j'}_k\leftarrow  0, \forall j' \in \I_{[1,L-1]}$, solve \eqref{eq:H_PCE_SOCP} \label{step:selectBegin} 
\If{\eqref{eq:H_PCE_SOCP} is feasible and \edit{$V^{\text{m}} _k\leq \tilde{J}_k$} }   $q \leftarrow 0$
\Else \hspace{0.2cm} backup initial condition: 
\State  $L \leftarrow L_x + (N+q) (L_w -1)$, $q \leftarrow q+1$
\State   $ \bar{\pce{x}}^j_{k}\leftarrow \pce{x}^{j,\star}\pred{k}{k-1},  \forall j \in \I_{[0,L-1]}$, solve \eqref{eq:H_PCE_SOCP} \edit{ and update $V_N^{\text{b}}$}
\EndIf \label{step:selectEnd}
\State \edit{$V_{N,k} \leftarrow \min\left\{V_N^{\text{m}},V_N^{\text{b}}\right\}$}
\State Apply $u^{\text{cl}}_k$, cf. \eqref{eq:scalarFeed} to system \eqref{eq:RVdynamics}, update $\tilde{J}_{k+1}$ and $\pce{x}^{j,\star}\pred{k+1}{k},  \forall j \in \I_{[0,L-1]}$ \label{step:MPC_2}
\State  $k\leftarrow k+1$, go to Step \ref{step:MPC_0}
\end{algorithmic}
\end{algorithm}

Now, combining with the results of previous sections, we propose a data-driven stochastic predictive control algorithm based on OCP~\eqref{eq:H_PCE_SOCP} as summarized in Algorithm~\ref{alg:datadrivenSMPC}. Note that in the data collection and pre-processing phase, random inputs $\trar{u}{T-1}$ are generated to obtain $\trar{x}{T}$, and  to estimate $\trar{w}{T-1}$. They are also used to calculate terminal ingredients as detailed in Remarks~\ref{remark:noiseEst}--\ref{remark:dataTerminal}, respectively. \edit{For the sake of sampling efficiency, we note that one might obtain  $\trar{u}{T-1}$ via closed-loop experiments.\cite{VanWaarde2021beyond,Iannelli2021}}

For the online optimization phase, we first assume OCP~\eqref{eq:H_PCE_SOCP} is feasible with the \edit{measured} initial condition \eqref{eq:initialxk} at time instant $k=0$. Then, we update polynomial bases and initial conditions for each time instant $k \in \N^+$, in Step~\ref{step:selectBegin}--\ref{step:selectEnd}. 
Precisely, inspired by the work of Farina et al.,\cite{Farina13p} we propose a selection strategy between the \edit{measured} initial condition \eqref{eq:initialxk} and the backup initial condition \eqref{eq:initialXk}.
\edit{Henceforth, we denote the optimal value functions~\eqref{eq:H_PCE_SOCP_obj} evaluated with the measured initial condition \eqref{eq:initialxk}, respectively, with the backup initial condition \eqref{eq:initialXk} as $V_{N}^{\text{m}}$ and $V_{N}^{\text{b}}$; they are initialized to be $+\infty$ at time instant $k$. In addition, we denote the cost by evaluating the shifted trajectories solution \eqref{eq:feasibleShift} in the objective \eqref{eq:H_PCE_SOCP_obj} as $\tilde{J}_k$.
First, we consider the measured initial condition \eqref{eq:initialxk}.
} 
In case of OCP~\eqref{eq:H_PCE_SOCP} is feasible and $V_{N}^{\text{m}} \leq \tilde{J}_k$, we accept the optimal solutions. Otherwise, we turn to the backup initial condition \eqref{eq:initialXk}. Doing so, OCP~\eqref{eq:H_PCE_SOCP} is ensured to be feasible as shown in Lemma \ref{lemma:recurFeasi} and \edit{$V_{N}^{\text{b}}\leq \tilde{J}_k$} holds due to optimality. Consequently, the proposed selection strategy of initial conditions implies the following 
\edit{\begin{equation}\label{eq:costDecayBackup}
		V_{N,k} \doteq \min\left\{ V_{N}^{\text{m}},V_{N}^{\text{b}} \right\}\leq \tilde{J}_k,\quad k\in \N
	\end{equation}
}
where $V_{N,k} $ denotes the optimal cost of \eqref{eq:H_PCE_SOCP} after the selection of initial conditions at current time instant $k$.

In addition, as shown in Lemma~\ref{lemma:recurFeasi}, the choice of backup initial condition leads to \edit{a} change of polynomial basis.
\begin{remarkmod}[Polynomial bases in closed loop] \label{rem:polynomials}
 At time instant $k$, the polynomial basis of OCP~\eqref{eq:H_PCE_SOCP} is constructed 
\begin{itemize}
\item[(i)] in case of  the \edit{measured} initial condition \eqref{eq:initialxk},  as
\begin{subequations}\label{eq:CLPoly}
\begin{equation}\label{eq:primePoly}
\{\phi^j\}_{j=0}^{L-1}=\{1, \{\phi_k^j\}_{j=1}^{L_w-1},\cdots, \{\phi_{k+N-1}^j\}_{j=1}^{L_w-1}\},
\end{equation}
 with $L = 1+ N(L_w-1)$ and $L_x = 1$,
\item[(ii)] and in case of the backup initial condition \eqref{eq:initialXk} consecutively $q \in \N^+$ times, 
\begin{equation}
\label{eq:backupPoly}
\{\phi^j\}_{j=0}^{L-1}=\{  1, \underbrace{ \{\phi_{k-q}^j\}_{j=1}^{L_w-1},\cdots,\{\phi_{k-1}^j\}_{j=1}^{L_w-1}}_{\{\phi_{x}^j\}_{j=1}^{L_x-1}}, \{\phi^j_{k}\}_{j=1}^{L_w-1} , \cdots \{\phi_{k+N-1}^j\}_{j=1}^{L_w-1}\}, 
\end{equation}
\end{subequations}
 with $ L = 1+ (N+q)(L_w-1)$ and $L_x = 1+q(L_w-1)$. 
\end{itemize} 
Note that if the controller selects backup initial conditions from time instant $k-q$ to $k$ consecutively, the basis keeps the original polynomials $\{1, \{\phi_{k-q}^j\}_{j=1}^{L_w-1},\cdots, \{\phi_{k-q+N-1}^j\}_{j=1}^{L_w-1}\}$ and new polynomials $\{\{\phi_{k-q+N}^j\}_{j=1}^{L_w-1}, \cdots, \{\phi_{k+N-1}^j\}_{j=1}^{L_w-1}\}$ are appended in the end. Moreover, the polynomial basis $\{\phi_x^j\}_{j=0}^{L_x-1}$ of $\bar{X}_k$ is  $\{1, \{\phi_{k-q}^j\}_{j=1}^{L_w-1},\cdots,\{\phi_{k-1}^j\}_{j=1}^{L_w-1}	\}$ with $L_x =1+ q (L_w-1)$, which also indicates $\bar{X}_k$ is related to the last $q$ process \edit{disturbance}s $W_{k-q}$ $\cdots$ $W_{k-1}$.  \End
\end{remarkmod}

\edit{After selecting initial conditions and solving OCP~\eqref{eq:H_PCE_SOCP}, we note that the solution to \eqref{eq:H_PCE_SOCP} gives the PCE coefficients of the first predicted input $\pce{u}^{j,\star}_{k|k}$ for $j \in I_{[0,L]}$ which can not be applied to the system \eqref{eq:RelizationDynamics}. To determine the feedback $u^{\text{cl}}_k$ to be applied to the realization dynamics \eqref{eq:RelizationDynamics}, consider 
	\begin{equation}\label{eq:scalarFeed}
		u^{\text{cl}}_k \doteq \sum_{j=0}^{L_x-1} \pcecoe{u}{j\star}_{k|k} \phi^j_x \relx.
	\end{equation}
	Note that $\pcecoe{u}{j\star}_{k|k} =0 $ holds for $j \in \I_{[L_x, L-1]}$ due to the causality condition \eqref{eq:causality}. The realizations of polynomials $  \{\phi_{x}^j \relx\}_{j=1}^{L_x-1}=\{ \{\phi_{k-q}\relx\}_{j=1}^{L_w-1},\cdots,\{\phi_k\relx\}_{j=1}^{L_w-1}, \{\phi_{k+1} \relx\}_{j=1}^{L_w-1}\}$ are fitted  a posteriori using the knowledge of the past $q$ \edit{disturbance} realizations. Specifically, we evaluate
	\begin{equation}\label{eq:estimatePhi}
		\min_{\substack{\phi_{k-i}^j\relx \in \R, i \in \I_{[1,q]}, j\in \I_{[1,L_w-1]} 
		}} \sum_{i=1}^q \left\| w_{k-i} -\sum_{j=1}^{L_w-1} \pce{w}^{j} \phi_{k-i}^j\relx\right \|^2 .
	\end{equation}
 Notice that \eqref{eq:estimatePhi} requires measuring/obtaining disturbance realizations. In the case of unmeasured disturbances, we estimate $w_{k-1}$ online by appending $(u_{k-1}, x_{k-1})$ to the offline input-state data $\begin{bmatrix}
 	\xxd \\
 	\uud
 \end{bmatrix}$ in \eqref{eq:estimdateWd}. That is, we employ~\eqref{eq:estimdateWd} with respect to the data matrix $\left[\begin{array}{c|c}
 \xxd & x_{k-1} \\
 \uud & u_{k-1}
\end{array}\right]$ and take the last element of the estimated \edit{disturbance}s as $w_{k-1}$.\cite{Pan21s} \vspace*{3mm}
}

We see that the feedback \eqref{eq:scalarFeed} takes the initial conditions in PCE coefficients and the past $q$ \edit{disturbance} realizations into account. Thus, it still reflects the information of the current state realization when the backup initial conditions are chosen; while Farina et al. \cite{Farina13p} neglect the current state measurement in this case.
While applying the feedback $u^{\text{cl}}_k$ to  the realization system \eqref{eq:RelizationDynamics}, we construct the shifted trajectories \eqref{eq:feasibleShift} based on the current optimal solutions to update $\tilde{J}_{k+1}$ and $\pce{x}^{j,\star}\pred{k+1}{k},  \forall j \in \I_{[0,L-1]}$ for next time instant $k+1$. Hence, due to the selection of initial conditions, the proposed predictive control scheme becomes a dynamic feedback strategy with  $\tilde{J}_{k+1}$ and $\pce{x}^{j,\star}\pred{k+1}{k}, \forall j \in \I_{[0,L-1]}$ as internal memory states of the controller.

\subsection{Closed-loop dynamics and properties}

Now, we formulate the following closed-loop dynamics \edit{resulting from the application of} Algorithm~\ref{alg:datadrivenSMPC}  with a specific initial state $x_{\text{ini}}$ and a \edit{disturbance} sequence $w_k$, $k\in \N$
\edit{\begin{subequations}\label{eq:CL}
		\begin{align}
			\label{eq:sotchasticCLRe}
			x_{k+1} = A x_{k} + B u^\text{cl}_k +w_k, \quad  x_0=x\inst{\text{ini}},\quad  k \in \N.
		\end{align}
Moreover, we obtain the sequence of optimal value functions $V_{N,k} \in \R$, $k \in \N$ with respect to the OCP~\eqref{eq:H_PCE_SOCP} evaluated with the selected initial condition in the closed loop.

 Similarly, with a probabilistic initial state $X_0=X\inst{\text{ini}}$ and the stochastic process disturbances $W_k$, we have the closed-loop dynamics in random variables,
\begin{equation}
\label{eq:sotchasticCL}
X_{k+1} = A X_{k} + B U^\text{cl}_k +W_k, \quad X_0=X\inst{\text{ini}}, \quad k \in \N,
\end{equation}
where conceptually the realization of $U^\text{cl}_k$, $U^\text{cl}_k\relx = u^\text{cl}_k$. In other words, \eqref{eq:sotchasticCLRe} is a realization of \eqref{eq:sotchasticCL} for the uncertainty outcome $\omega \in \Omega$. Similarly, the closed-loop evolution of the optimal value function $V_{N,k}$, $k \in \N$ can be regarded as the realization of the underlying random variable $\mcl{V}_{N,k} \in \spl$ such that $V_{N,k}=\mcl{V}_{N,k} \relx$ for all $k \in \N$. 
\end{subequations}
}
The following proposition establishes recursive feasibility and an asymptotic average cost of the proposed data-driven stochastic predictive control method based on Algorithm \ref{alg:datadrivenSMPC} and OCP \eqref{eq:H_PCE_SOCP}.
\begin{theorem}[Recursive feasibility and average cost bound]\label{thm:costDecay}
 Consider the closed-loop dynamics \eqref{eq:CL} determined by the data-driven stochastic predictive control method based on Algorithm \ref{alg:datadrivenSMPC} and OCP \eqref{eq:H_PCE_SOCP}. Let Assumptions \ref{ass:AB}-- \ref{ass:finitenoise} hold. Consider  $x_{\text{ini}}$ as a  realization of  $X_{\text{ini}}$. Suppose at time instant $k=0$, OCP~\eqref{eq:H_PCE_SOCP} is feasible with the initial condition $\pce{x}_0^{0}=x_{\text{ini}}$, $\pce{x}_0^{j}=0$, $\forall j \in \I_{[0,L-1]}$ with $L = 1+N(L_w-1)$ dimensional polynomial basis  as described in \eqref{eq:primePoly}, then
 \begin{itemize} \item[(i)]
  OCP \eqref{eq:H_PCE_SOCP} is feasible at time instant $k \in \N$ with the selected initial condition $\bar{\pce{x}}^{[0,L-1]}_k$ and updated $L$ dimensional polynomial basis as shown in \eqref{eq:CLPoly}.
 \item[(ii)] \edit{Moreover,  OCP \eqref{eq:H_PCE_SOCP} at consecutive time instants satisfies
 	\begin{equation}\label{eq:costDecay}
 		\mean\big[\mathcal{V}_{N,k+1} - \mcl{V}_{N,k}\big] \leq -  \mean\big[\|X_k\|^2_Q + \|U^\text{cl}_k\|^2_R\big] + \alpha,
 \end{equation}
	with $\alpha\doteq \trace\left( \covar\left[W\right] P \right)$.}
\item[(iii)] \edit{In addition, the average asymptotic cost of the proposed data-driven stochastic predictive control method is bounded from above by $\alpha \in \R^+$,
	\begin{equation}\label{eq:averageCostBound}
		\lim_{k\to \infty} \frac{1}{k} \sum_{i=0}^k
		\mean\big[\|X_i\|^2_Q + \|U^\text{cl}_i\|^2_R\big] \leq \alpha.
	\end{equation}}
\end{itemize}
\end{theorem}
\begin{proof}
First, the statement on recursive feasibility, i.e. Statement (i), directly follows from Lemma \ref{lemma:recurFeasi}. 
Then, with the selection strategy proposed in Algorithm \ref{alg:datadrivenSMPC}, we have condition \eqref{eq:costDecayBackup} be satisfied. By applying it to the left hand side of \eqref{eq:costDecay}, we have
\begin{equation*}%\label{eq:proofVk}
\begin{aligned}
 V_{N,k+1}& - V_{N,k} \leq  \tilde{J}_{k+1}- V_{N,k}   \\
= & \sum_{j=0}^{L-1} \Big( \sum_{i=k+1}^{k+N-1} \big(\| \pcecoe{x}{j,\star} \pred{i}{k}\|^2_Q  +\|\pcecoe{u}{j,\star}\pred{i}{k}\|^2_R \big)+  \|\pce{x}^{j,\star}_{k+N|k}\|_Q^2 +  \|K\pce{x}^{j,\star}_{k+N|k}\|_R^2  +\|M\pce{x}^{j,\star}_{k+N|k} \|_P^2 \Big)\langle \phi^j,\phi^j\rangle +\sum_{j=0}^{L_w-1} \|\pce{w}^j_{k+N}\|_P^2 \langle \phi^j_{k+N},\phi^j_{k+N}\rangle   \\
& -  \sum_{j=0}^{L-1} \Big(\| \bar{\pce{x}}^{j}_{k}\|^2_Q  +\|\pcecoe{u}{j,\star}_{k|k}\|^2_R +  \sum_{i=k+1}^{k+N-1}  \big(\| \pcecoe{x}{j,\star} \pred{i}{k}\|^2_Q  +\|\pcecoe{u}{j,\star}\pred{i}{k}\|^2_R\big)+ \|\pce{x}^{j,\star}_{k+N|k}\|_P^2 \Big)\langle \phi^j,\phi^j\rangle\\
= & \sum_{j=0}^{L-1} \Big( -\|  \bar{\pce{x}}^{j}_{k}\|^2_Q  -\|\pcecoe{u}{j,\star}_{k|k}\|^2_R +
\pce{x}^{j,\star \top}_{k+N|k}
\underbrace{(H^\top M^\top  P M H -P +Q + K^\top R K )}_{=0 ~\text{due to }\eqref{eq:dataRiccati}} \pce{x}^{j,\star}_{k+N|k}  \Big)\langle \phi^j,\phi^j\rangle  +\sum_{j=0}^{L_w-1} \|\pce{w}^j_{k+N}\|_P^2 \langle \phi^j_{k+N},\phi^j_{k+N}\rangle.
\end{aligned}
\end{equation*}

\edit{Thus, we have 
\[
V_{N,k+1} - V_{N,k}  \leq \sum_{j=0}^{L-1} \left( -\|  \bar{\pce{x}}^{j}_{k}\|^2_Q  -\|\pcecoe{u}{j,\star}_{k|k}\|^2_R\right)  \langle \phi^j,\phi^j\rangle +\sum_{j=0}^{L_w-1} \|\pce{w}^j_{k+N}\|_P^2 \langle \phi^j_{k+N},\phi^j_{k+N}\rangle.
\]
for the  selected initial condition $\bar{\pce{x}}^{[0,L-1]}_k$ at time instant $k$. Generalizing the above to all uncertain outcomes of \eqref{eq:sotchasticCL}, we have
\[
\mean\left[\mcl{V}_{N,k+1} - \mcl{V}_{N,k} \right] \leq -\mean\big[\|X_k\|^2_Q + \|U^{\text{cl}}_k\|^2_R\big]  + \mean \big[\|W_{k+N}\|^2_P\big].
\]
\edit{Since $W_k$ are} i.i.d for $k \in \N$, we have $\mean \big[\|W_{k+N}\|^2_P\big] = \trace(\covar[W]P)$. 
}

Furthermore, following the typical argument in stochastic MPC,\cite{Cannon09p} we see that \eqref{eq:costDecay} implies
\begin{align*}
0\leq \lim_{k\to \infty}  \frac{1}{k} \left(\mean\left[\mcl{V}_{N,k} - \mcl{V}_{N,0} \right]\right) 
&\leq  \lim_{k\to \infty} \frac{1}{k}  \sum_{i=0}^k \Big( - \mean\big[\|X_i\|^2_Q + \|U_i^{\text{cl}}\|^2_R\big] +  \mean\big[\|W_{i+N}\|^2_P\big]\Big).
\end{align*}
With $ \alpha =\lim_{k\to \infty} \frac{1}{k}  \sum_{i=0}^k  \mean\big[\|W_{i+N}\|^2_P\big] = \trace\left(\covar\left[W\right]P\right)  $ statement (iii) follows.
\end{proof}

The average asymptotic cost bound in \eqref{eq:averageCostBound} provides a notion of stability and convergence for the stochastic closed-loop dynamics.\cite{Kouvaritakis10e,Cannon09p} We remark that,  due to possible vary large \edit{disturbance} realizations, the convergence \eqref{eq:averageCostBound} does not hold for one specific realization of the closed-loop dynamics \eqref{eq:sotchasticCLRe}. Rather it holds for the random-variable closed-loop dynamics \eqref{eq:sotchasticCL} with its average stage cost evaluated at the current time instant.  
It remains to illustrate the   average asymptotic cost bound \eqref{eq:averageCostBound} and the closed-loop properties of the proposed data-driven stochastic predictive control method 
 in the next section.

\section{Simulation Examples}\label{sec:simulation}
\edit{In this section, we consider two examples to illustrate the efficacy and the closed-loop properties of the proposed data-driven predictive control schemes, cf. Algorithm~\ref{alg:datadrivenSMPC}. In the first (scalar) example, we consider two cases with different distributions of the initial state and the process disturbances under the assumption of exactly measured disturbances. The second example considers a batch reactor system with two inputs and four states. Therein, we compare three variants of Algorithm~\ref{alg:datadrivenSMPC}, i.e.  data-driven variants with measured/estimated disturbances and a model-based variant with a model identified from the input-state data.}
 
\subsection{Scalar dynamics}
We consider the following scalar stochastic system
\begin{equation}\label{eq:sclarsys}
X\inst{k+1} = 2X\inst{k} + U\inst{k} + W\inst{k}, \quad X_0 = X\ini, \quad k \in \set{N}.
\end{equation}
This scalar system is originally used in Gr{\"u}ne~\cite{grune13economic} and then in Ou et al.~\cite{Ou21} 

In  OCP \eqref{eq:H_PCE_SOCP} we consider to stabilize the system \eqref{eq:sclarsys} with the prediction horizon $N = 25$ and the penalty matrices $Q=R=1$. The chance constraints  \eqref{eq:chance} are chosen as $\prob[X\in \set{X}]\leq 1-\varepsilon_x$  and $\prob[U\in \set{U}]\leq 1-\varepsilon_u$ with $\set{X}=[-2,2]$, $\set{U}=[-3,3]$ and $\varepsilon_x=\varepsilon_u=0.1$. Two different settings of the initial state and process \edit{disturbance} are investigated:
\begin{itemize}
	\item Case 1: uniformly distributed $X\ini \sim \mcl{U}(0,2)$ and Gaussian distributed $W_k \sim \mcl{N}(0,0.1^2)$ \label{item:1}
	\item Case 2: beta distributed $X\ini \sim \mcl{B}(0.5,0.5)$ and uniformly distributed $W_k \sim \mcl{U}(-0.173,0.173)$.
\end{itemize}
Note both of the initial states have compact supports and satisfy the state chance constraint. Moreover, the two distributions of \edit{disturbance} have zero-means and the same variance, i.e., $\mbb{V}[W_k]=0.01$. The above settings ensure the feasibility at time instant $k=0$ and consequently also the recursive feasibility , cf the input policy of Algorithm \ref{alg:datadrivenSMPC}.

Next, following the data collection and pre-processing phase of Algorithm \ref{alg:datadrivenSMPC} we choose $T=100$ and record state-input-\edit{disturbance} measurements of system \eqref{eq:sclarsys} for both cases. Note the chosen $\trar{(u,w)}{T-1}$ are ensured to be persistently of exciting of order more than $N+\dimx + 1 = 27$. In addition, we assume \edit{disturbance} realizations are exactly measurable here.

Then, with the recorded data we construct the Hankel matrices in \eqref{eq:H_PCE_SOCP_hankel} and get $P=4.236$, $K = -1.618$ and $\Gamma=0.0117$  for the terminal ingredients by solving \eqref{eq:dataPKG} and \eqref{eq:gamma}. In the online optimization phase, we consider that we have the exact measurement of the current state $x_k$, $k\in \N$ and update the initial conditions as well as the polynomial basis accordingly with respect to Algorithm \ref{alg:datadrivenSMPC} and \eqref{eq:CLPoly}.

\subsubsection*{Case 1: Uniformly distributed initial state and Gaussian distributed \edit{disturbance}} \label{sec:simu_Gauss}
Suppose the initial state $X\ini$ follows uniform distribution $\mcl{U}(0,2)$ and the process \edit{disturbance} $W_k$ is Gaussian $\mcl{N}(0,0.1^2)$ for all $k \in \N$. Then, by choosing $\sigma(\varepsilon_x)$ and $\sigma(\varepsilon_u)$ appropriately according to the Gaussian cumulative distribution function, we obtain equivalent data-driven chance constraints  \eqref{eq:PCE_chance} with $\sigma(\varepsilon_x)= \sigma(\varepsilon_u) =1.645$ for $\varepsilon_x=\varepsilon_u=0.1$. Additionally, referring to Table \ref{tab:askey_scheme}, Assumption~\ref{ass:finitenoise} holds with $L_w=2$ when a Hermite polynomial basis is chosen for $W\inst{k}$. Note that we have different polynomial bases for OCPs at different time instants $k$, cf. \eqref{eq:CLPoly}. For the OCPs that follow \eqref{eq:initialxk}, i.e., directly use measured $x_k$ as initial condition, the polynomial basis are in $L=N(L_w-1)+1=26$ terms with $L_x=1$, cf. Lemma \ref{lem:no_truncation_error}. In contrast, for the OCPs with backup initial conditions \eqref{eq:initialXk}, $L_x$ is equal to the number of steps when backup initial conditions are selected consecutively, i.e.  $q$ in \eqref{eq:CLPoly}.

With $10$ realizations of the \edit{disturbance} sequence and initial state, the left subplots of Figure~\ref{fig:scalar_responses}~(a) display the corresponding closed-loop realization trajectories of system \eqref{eq:sclarsys} driven by Algorithm \ref{alg:datadrivenSMPC}. We see that the realizations of states and inputs rarely violate the system constraints and thus indicate the recursive feasibility. Moreover, the state realizations stay in a neighborhood of the origin after the first approaching phase. Moreover, Figure~\ref{fig:scalar_responses}~(a) depicts the averages of the state and input realization trajectories as well as the average asymptotic cost for a total of 1000 samplings. We denote $\bar{\ell}_k =  10^{-3} \cdot \frac{1}{k}\sum_{m=1}^{1000} \sum_{i=0}^k
\big[\|x_i^m\|^2_Q + \|u_i^m\|^2_R\big]$ as the statistical estimation of the average asymptotic cost in \eqref{eq:averageCostBound} with $(\cdot)^m$ as the $m$-th sample trajectory. Then, as shown in Figure~\ref{fig:scalar_responses}~(a), the stability in the notion of average asymptotic cost bound \eqref{eq:averageCostBound} is clearly illustrated. Furthermore, the realization trajectories of state and input are depicted in the $x-u$ plane in right subplot of Figure~\ref{fig:scalar_responses}~(a) for all of the samples. As one can see, the trajectories are centred at the line $u=Kx$, which suggests the chance constraints are seldom activated when we solve OCP \eqref{eq:H_PCE_SOCP} at each time instant for case 1. An intuitive explanation for this phenomenon is that the Gaussian distributed \edit{disturbance}  has a small variance and zero mean and therefore, the chance constraints are satisfied with rather high probability for the input policy $u=Kx$. Additionally, the evolution of the normalized histograms of the state realizations at time instant $k=0,5,10,15,20$ is shown in Figure~\ref{fig:scalar_hist}~(a), where the vertical axis refers to probability density. It can be seen that the distribution of $X$ keeps stationary over time after the first approaching stage.

\begin{figure}[t!]
	\centering
	\begin{tabular}[b]{c}
		\includegraphics[width=.8\linewidth]{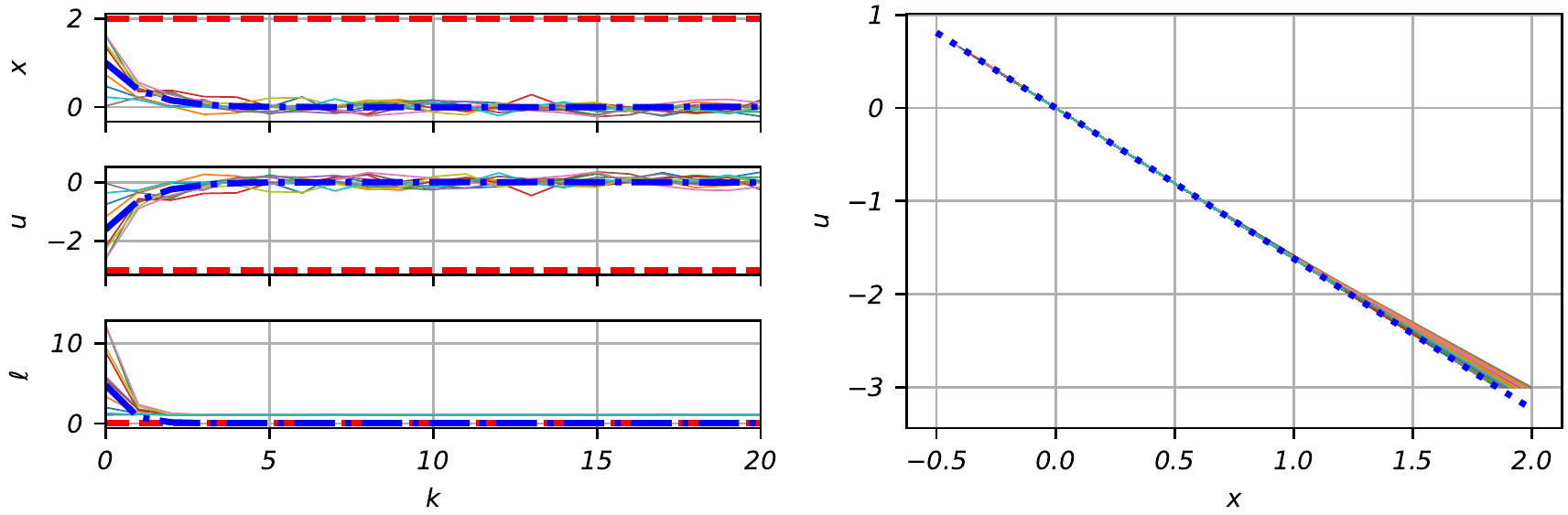} \\
		(a) \\
		\includegraphics[width=.8\linewidth]{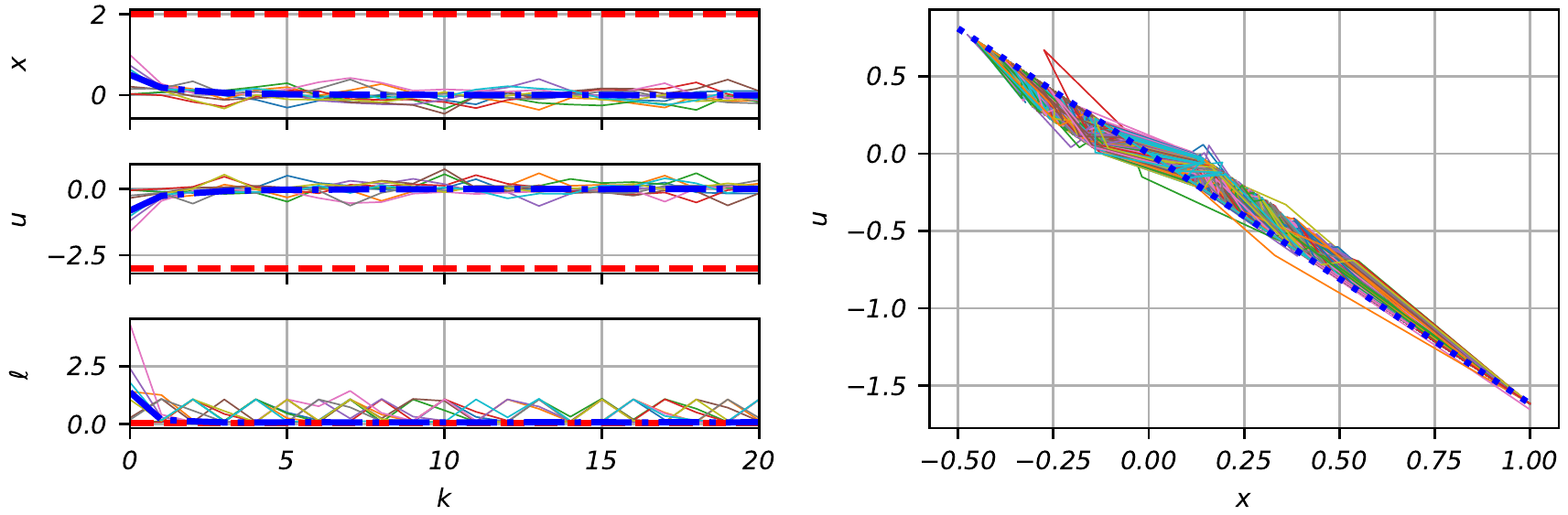} \\
		(b)\\
		\includegraphics[width=.8\linewidth]{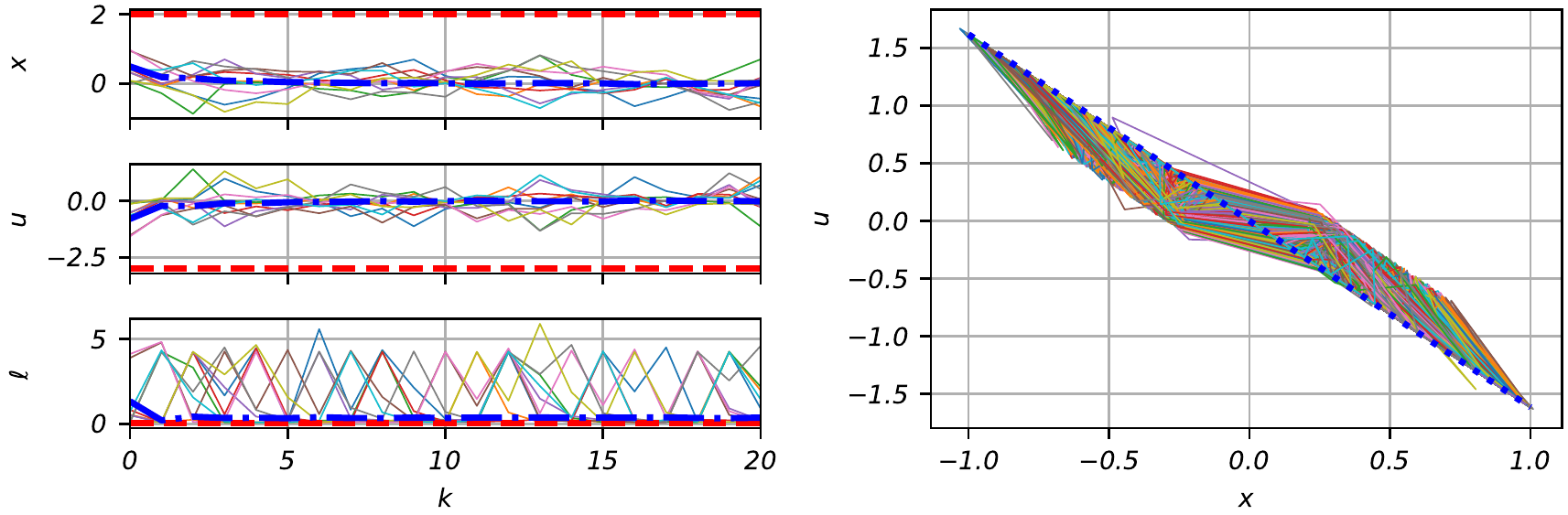} \\
		(c)
	\end{tabular}
	\caption{10 different closed-loop data-driven MPC realizations and the average of 1000 realizations for three cases. (a)~Case~1. (b)~Case~2. (C)~Case~2  with $W_k\sim\mcl{U}(0.346,0.346)$.}
	\label{fig:scalar_responses} 
\end{figure}
\begin{figure}[tb!]
	\centering
	\begin{tabular}[b]{c}
		\includegraphics[width=.4\linewidth]{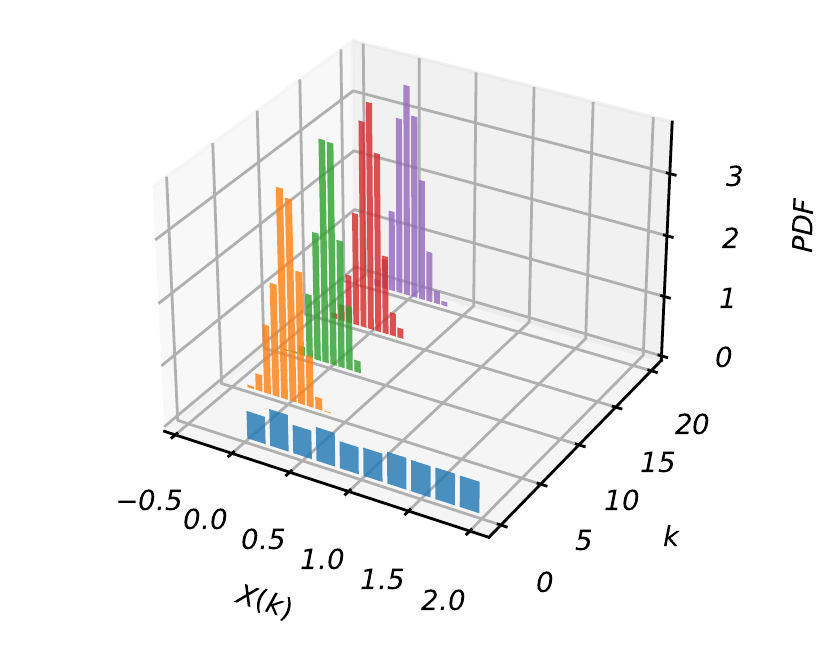} \\
		(a)
	\end{tabular} 
	\begin{tabular}[b]{c}
		\includegraphics[width=.4\linewidth]{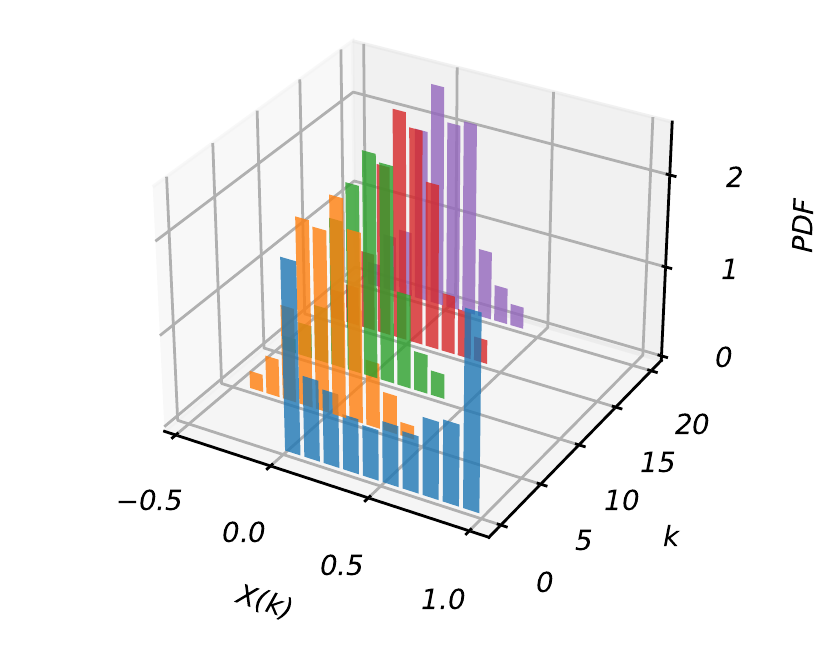} \\
		(b)
	\end{tabular}
	\caption{Histograms of the state from 1000 closed-loop realizations for two cases. (a)~Case~1. (b)~Case~2.}
	\label{fig:scalar_hist} 
\end{figure}

\subsubsection*{Case 2: Beta distributed initial state and uniformly distributed \edit{disturbance}}
Suppose the initial state $X\ini$ follows beta distribution $\mcl{B}(0.5,0.5)$ and the process \edit{disturbance} $W_k$ is uniformly distributed as $\mcl{U}(-0.173,0.173)$ for all $k \in \N$. In terms of a conservative choice, we obtain  $\sigma(\varepsilon_x)= \sigma(\varepsilon_u) =\sqrt{(2-0.1)/0.1} = 4.359$ for chance constraints \eqref{eq:PCE_chance} with $\varepsilon_x=\varepsilon_u=0.1$. Similarly, referring to Table \ref{tab:askey_scheme}, the uniformly distributed $W_k$ admits an exact PCE description with $L_w=2$ in the basis of Legendre polynomials. Therefore, we conclude the same choice of $L$ as Section \ref{sec:simu_Gauss}. For the above settings, the stability in the sense of the average of realizations and distribution is clearly shown in Figure~\ref{fig:scalar_responses}~(b) and Figure~\ref{fig:scalar_hist}~(b). Note the chance constraints are more likely to be activated comparing to case 1 due to the conservative formulation of chance constraints when \edit{disturbance} is uniformly distributed. Therefore, in the right subplot of Figure~\ref{fig:scalar_responses}~(b), more deviations from the line $u=Kx$ are observed for the realization trajectories of state and input in the $x-u$ plane.

Furthermore, we increase the uncertainty of the process \edit{disturbance} by doubling its support, i.e., $W_k\sim\mcl{U}(0.346,0.346)$ and $\mbb{V}[W_k]=0.04$. Hence, the chance constraints are more likely to be activated and we obtain the solutions for 1000 samplings as shown in Figure~\ref{fig:scalar_responses}~(c). Note that in the right subplot, the realization trajectories of state and input that are more widely distributed around the feedback law $u=Kx$ are observed.

\subsection{Batch reactor example}
%\color{red}
As a second example, we consider a discretized batch reactor system used by De Persis and Tesi.~\cite{DePersis19} The system matrices which serve as simulated reality are
\[	
		A = \begin{bmatrix} \phantom{-}1.178 & \phantom{-}0.001 & \phantom{-}0.511 & -0.403 \\
			-0.051 & \phantom{-}0.661 & -0.011 & \phantom{-}0.061 \\
			\phantom{-}0.076 & \phantom{-}0.335 & \phantom{-}0.560 & \phantom{-}0.382 \\
			\phantom{-}0\phantom{.000} & \phantom{-}0.335 & \phantom{-}0.089 & \phantom{-}0.849\\
		\end{bmatrix}, \qquad
		B = \begin{bmatrix} \phantom{-}0.004 & -0.087 \\  \phantom{-}0.467 & \phantom{-}0.001 \\ \phantom{-}0.213 & -0.235 \\  \phantom{-}0.213 & -0.016 \end{bmatrix}.
\]
 We consider stabilizing the system at the origin with the prediction horizon $N = 10$ and the penalty matrices $Q=I_{\dimx}$ and $R= I_{\dimu}$.
  We suppose that the i.i.d. disturbances $W_k$ follow a Gaussian distribution with zero mean and covariance $\Sigma_W = 10^{-4}\cdot I_{\dimx}$.  Let $U^i$ denote the $i$-th component of $U$. We impose the chance constraint~\eqref{eq:chance} on $U^1$ with $\prob[U^1\in \set{U}]\leq 1-\varepsilon_u$, $\set{U}=[-2,2]$ and $\varepsilon_u=0.1$. Correspondingly, we find $\sigma(\varepsilon_u)=1.645$ according to the standard normal table.
   Moreover, we specify the terminal region $\Xf$ in \eqref{eq:PCE_Xf} as $\Xf \doteq \{x\in \set{X}| \|x\|_P^2\leq 10^{-2}\}$.  
 
 We remark that the least-square estimate~\eqref{eq:estimdateWd} of the past disturbance realizations implicitly specify an LTI system with 
 \begin{equation}\label{eq:estimatedAB}
 	\left[\hat{A} \, \hat{B}\right]  = \xxdp\begin{bmatrix}
 		\xxd \\
 		\uud
 	\end{bmatrix}^\rinv,
 \end{equation}
where $\mbf{x}_{[0,T]}$ and $\mbf{u}_{[0,T-1]}$ are the recorded state-input trajectories as specified in Remark~\ref{remark:noiseEst}.

In the following, we compare three variants of Algorithm~\ref{alg:datadrivenSMPC}:
 \begin{enumerate}[label=\Roman*]
 	\item Algorithm~\ref{alg:datadrivenSMPC} with exactly measured disturbances
 	 \label{sch:mea}
 	\item Algorithm~\ref{alg:datadrivenSMPC} with disturbances estimated from~\eqref{eq:estimdateWd}
 	 \label{sch:est}
 	\item Algorithm~\ref{alg:datadrivenSMPC} with identified model, i.e., we replace the non-parametric system representation~\eqref{eq:H_PCE_SOCP_hankel} in OCP~\eqref{eq:H_PCE_SOCP} by the PCE coefficient dynamics~\eqref{eq:PCEcoesDynamics} with the identified $\hat{A}$ and $\hat{B}$ from \eqref{eq:estimatedAB}. \label{sch:estAB}
 \end{enumerate}

In the offline data collection phase, we generate the data, i.e. $\mbf{x}_{[0,T]}$ and $\mbf{u}_{[0,T-1]}$, to estimate the disturbances  $\mbf{w}_{[0,T-1]}$ in Variant~\ref{sch:est}, to identify $\hat{A}$ and $\hat{B}$ in Variant~\ref{sch:est}, and to determine the terminal ingredients for all variants. Since the system to be controlled is open-loop unstable, we first apply 20  random samples $u_k$ from $ \mcl{U}(-1,1)^2$ to system~\eqref{eq:RelizationDynamics}. Then, we use the recorded state-input trajectory to estimate the 20 disturbances by~\eqref{eq:estimdateWd}. With the estimated/measured disturbances, we determine a feedback $\tilde K$ from~\eqref{eq:dataPKG} to pre-stabilize the system for further data collection. We continue to sample $980$ inputs with $u_k = \tilde K x_k + v_k$ and $v_k$ sampled from $\mcl{U}(-1,1)^2$ to estimate disturbances by~\eqref{eq:estimdateWd} for Variant~\ref{sch:est} and to identify $\hat{A}$ and $\hat{B}$ by~\eqref{eq:estimatedAB} for Variant~\ref{sch:estAB}. Finally, we use the first $120$ measured inputs, states, and measured/estimated disturbances to construct Hankel matrices for Variants~\ref{sch:mea} and~\ref{sch:est}. The terminal ingredients $K$, $P$ and $\Gamma$ of Variants~\ref{sch:mea} and~\ref{sch:est} are determined from \eqref{eq:dataPKG} and \eqref{eq:gamma} with measured or estimated disturbances, respectively, and Variant~\ref{sch:estAB} use the same as Variant~\ref{sch:est}.

Considering 10 initial conditions $x\ini$ sampled from $ \mcl{U}(-5,5)^4$ and 10 sequences of disturbance realizations, we compute the closed-loop responses of Variants \ref{sch:mea}--\ref{sch:estAB} for 30 time steps. Moreover, we record the computation time of performing Algorithm~\ref{alg:datadrivenSMPC} for each step and evaluate the closed-loop costs $J^{\text{cl}} \doteq \sum_{k=0}^{29}\|x_k\|_Q^2+\|u_k^{\text{cl}}\|_R^2$ for the resulting closed-loop trajectories. Table~\ref{tab:ChemComparison} summarizes the mean and the Standard Deviation (SD) of the computation time per step and of the closed-loop costs for the three variants.
As illustrated in Table~\ref{tab:ChemComparison}, the difference of the resulting closed-loop costs is minuscule for this specific example. However, Variant~\ref{sch:estAB} with estimated system matrices requires much less computation time than the data-driven variants which illustrates the need to develop tailored numerical schemes. One possible approach is to truncate the PCE series. In an illustrative example~\cite[Section V.B]{Pan21s}, we considered a truncation of PCEs from dimension $L=31$ to $L=4$, which halves the computation time with a closed-loop performance loss of $4.38\%$. However, we note that the rigorous analysis of closed-loop robustness with respect to  truncated PCEs is still open. In addition, the computation effort can be reduced via  multiple shooting reformulation~\cite{Ou23} and preconditioning of the Hankel matrices. Since the numerical aspects are beyond the scope of this paper, we leave the design of tailored numerical schemes with robustness guarantees for future work.

 Specifically, Figure~\ref{fig:ChemComparison}~(a) depicts  the input and system state response trajectories for the three variants of Algorithm~\ref{alg:datadrivenSMPC} with respect to  $x\ini=[-2.64,-1.53,-1.87,-4.92]^\top$. As one can see, the response trajectories are almost identical. In addition, Figure~\ref{fig:ChemComparison}~(b) shows the closed-loop evolutions of the infinity norm of resulting slack variable $\|c\|_{\infty}$  for Variants~\ref{sch:mea} and~\ref{sch:est}. In line with Remark~\ref{rem:slack}, the slack variable $c$ is with negligible magnitude and only serves the purpose of the numerical feasibility of~\eqref{eq:H_PCE_SOCP_hankel}.

\begin{table}[t!]
	\caption{Comparison of the computation times in \texttt{julia}  and the closed loop costs for 10 closed-loop realizations. }
	\label{tab:ChemComparison}
	\centering
		\begin{tabular}{ccccc}
			\toprule
			\multirow{2}{*}{Variant} &  \multicolumn{2}{c}{Computation time}  & 	\multicolumn{2}{c}{Closed-loop cost $[-]$}\\
			\cmidrule(lr){2-3} 	\cmidrule(lr){4-5}   & Mean $\SI{}{[s]}$ & SD $\SI{}{[s]}$  & Mean $[-]$ & SD $[-]$\\
			%	Variant&   [s] & Closed-loop cost $[-]$\\
			\midrule
			\ref{sch:mea} & 8.03 &3.98 & 276.7121 & 142.9663 \\
			\ref{sch:est} & 8.45 &3.64& 276.7104&142.9650\\
			\ref{sch:estAB} &1.93 &0.86 & 276.7118&142.9636\\ 
			\bottomrule
		\end{tabular}
\end{table}

\begin{figure}[t!]
	\centering
	\begin{tabular}[b]{c}
		\includegraphics[width=.8\linewidth]{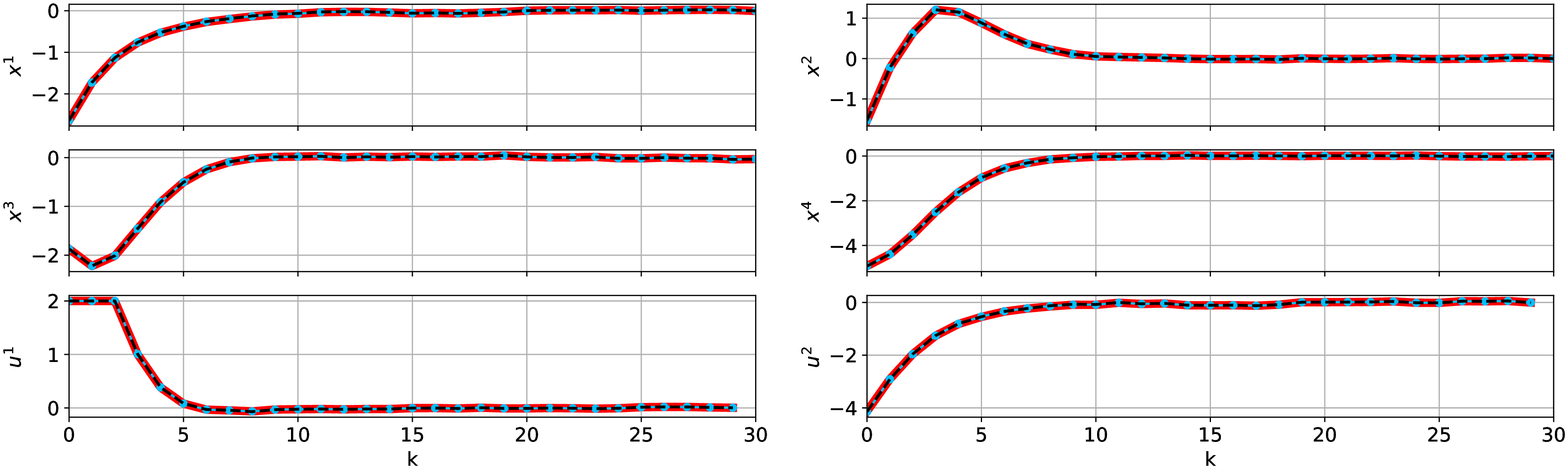}  \\
		(a) \\
			\includegraphics[width=.8\linewidth]{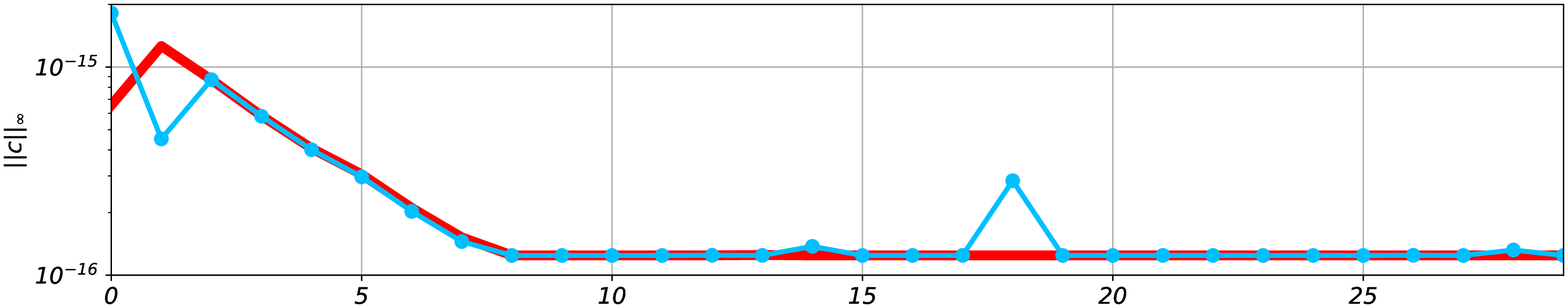}  \\
(b) 
	\end{tabular} 
	\caption{ Closed-loop trajectories of the batch reactor example with Gaussian disturbances for Variants \ref{sch:mea}--\ref{sch:estAB}. Initial condition $x\ini=[-2.64,-1.53,-1.87,-4.92]^\top$. Red-solid line: Variant~\ref{sch:mea}; blue-solid line with circle marker: Variant~\ref{sch:est}; black-dashed line: Variant~\ref{sch:estAB}. }\label{fig:ChemComparison} 
\end{figure}

Finally, we sample a total of $500$ sequences of disturbance realizations and initial conditions from $ \mcl{U}(-5,5)^4$. Then we compute the corresponding closed-loop responses of Variant~\ref{sch:mea}. The time evolution of the (normalized) histograms of the state realizations $x^1$ at $k=0, 5, 10, 15, 20, 25, 30$ is shown in Figure~\ref{fig:chem_hist}, where the vertical axis refers to the (approximated) probability density of $X^1$. As one can see, the proposed control scheme achieves a narrow distribution of~$X^1$ around $0$.

\begin{figure}[t!]
	\centering
		\includegraphics[width=.4\linewidth]{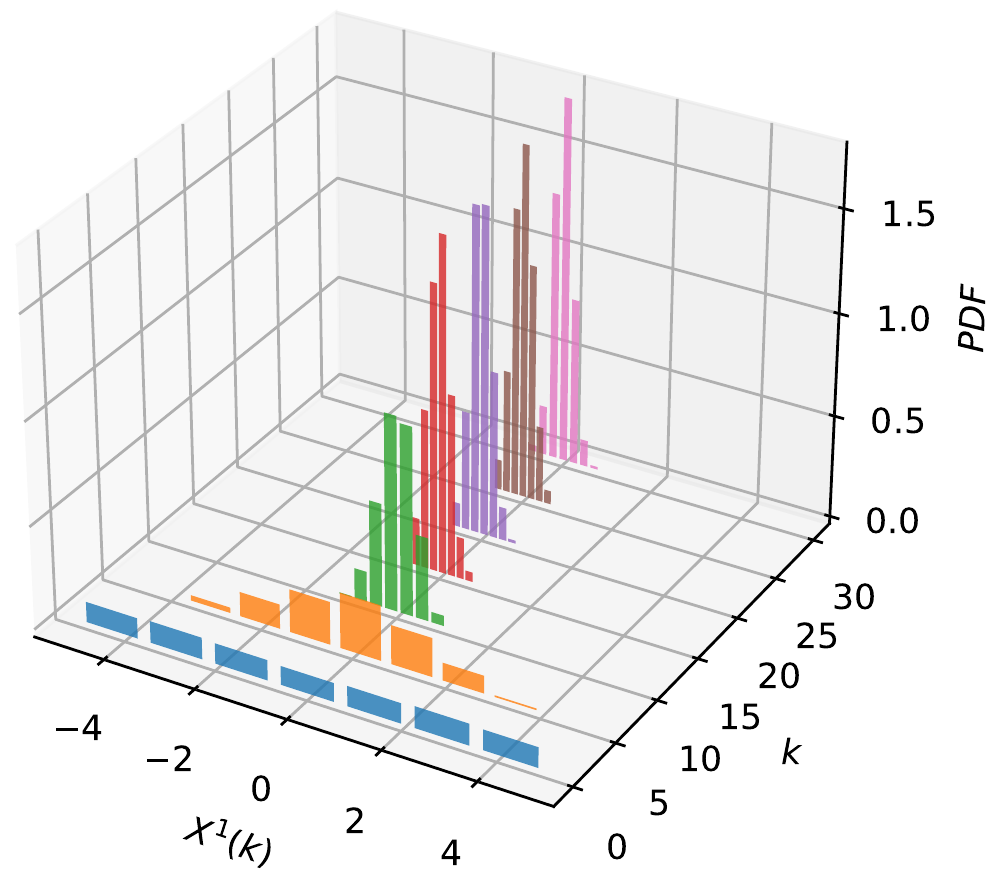}
	\caption{Histogram of $X^1$ from 500 closed-loop realizations for the batch reactor example.}
	\label{fig:chem_hist} 
\end{figure}
%\color{black}

\section{Conclusion and outlook}\label{sec:conclusion}
This paper has presented a data-driven stochastic predictive control scheme for stochastic LTI systems with general zero-mean process \edit{disturbances}, which uses past %$\{$state,input,noise$\}$ 
\edit{state, input, and disturbance}
realization trajectories to predict future PCE coefficients. This enables a formulation of the predictive control problem with data-driven uncertainty quantification and propagation via
polynomial chaos expansions. Moreover, we considered terminal constraints on the mean value and covariance of the terminal state. Combined with the idea of using backup initial conditions, we proposed a data-driven stochastic predictive control approach with recursive feasibility guarantee and further more with stability property in the sense of asymptotic average cost bound. The simulation results have illustrated the efficacy and closed-loop properties of the proposed scheme. % by considering two cases of different uncertain initial states and process \edit{disturbance}. 
 \edit{Future work should consider (i) the design and analysis of data-driven stochastic output-feedback predictive control,\footnote{In a follow-up to this paper we provide first results in this direction.\cite{Pan2022}} (ii)  truncation or approximation strategies of the polynomial basis with robustness guarantee, e.g. relying on sparse PCEs\cite{Luethen2021a} or on active subspaces,\cite{Pan2023} (iii) the robustness analysis with respect to errors in the estimation of past disturbance realizations, and (iv) the real-world application with tailored numerical algorithms.}

\bibliographystyle{IEEEtran}
\bibliography{Stochastic}
\end{document}